\documentclass[pra,aps,superscriptaddress,onecolumn,nofootinbib,longbibliography,accepted=2022-10-06]{quantumarticle}

\usepackage[english]{babel}

\newcommand{\pmaz}[1]{ \textcolor{black}{#1}}

\usepackage{amsmath}

\usepackage{graphicx}
\usepackage[colorlinks=true, allcolors=blue]{hyperref}
\newcommand{\Tr}{\mathrm{Tr}}
\renewcommand{\v}[1]{\ensuremath{\boldsymbol #1}}

\usepackage{amsthm}

\newtheorem{thm}{Theorem}
\newtheorem{lemma}[thm]{Lemma}
\newtheorem{cor}[thm]{Corollary}

\newtheorem{prop}[thm]{Proposition}
\newtheorem{defn}[thm]{Definition}
\usepackage{physics}
\usepackage{amsfonts}
\usepackage{authblk}

\title{Extraction of ergotropy: free energy bound and application to open cycle engines}

\author[1,]{Tanmoy Biswas}
\author[1,]{Marcin {\L}obejko}
\author[1]{Pawe{\l} Mazurek}
\author[3]{Konrad Ja{\l}owiecki}
\author[1]{Micha{\l} Horodecki}

\affil[1]{International Centre for Theory of Quantum Technologies, University of Gdansk, Wita Stwosza 63, 80-308 Gdansk, Poland}
\affil[3]{Institute of Theoretical and Applied Informatics, Polish Academy of Sciences, Ba{\l}tycka 5, 44-100 Gliwice, Poland}

\begin{document}
\maketitle

\begin{abstract}

  The second law of thermodynamics uses change in free energy of macroscopic systems to set a bound on performed work. Ergotropy plays a similar role in microscopic scenarios, and is defined as the maximum amount of energy that can be extracted from a system by a unitary operation. In this analysis, we quantify how much ergotropy can be induced on a system as a result of  system's interaction with a thermal bath, with a perspective of using it as a source of work performed by microscopic machines. We provide the fundamental bound on the amount of ergotropy which can be extracted from environment in this way. The bound is expressed in terms of the non-equilibrium free energy difference and can be saturated in the limit of infinite dimension of the system's Hamiltonian. The ergotropy extraction process leading to this saturation is numerically analyzed for finite dimensional systems. Furthermore, we apply the idea of extraction of ergotropy from environment in a design of a new class of stroke heat engines, which we label open-cycle engines. Efficiency and work production of these machines can be completely optimized for systems of dimensions 2 and 3, and numerical analysis is provided for higher dimensions.

\end{abstract}

\section{Introduction}

The second law of thermodynamics governs what state transformations are allowed if work and heat are exchanged with the environment. These transformations are described in terms of the thermodynamic potentials, which, regardless of the microscopic details, provides fundamental limitations for energy transformations in thermodynamic processes. In particular, these limitations set the upper limits on power and efficiency. \pmaz{One of the most important thermodynamic potentials is the free energy, whose non-equilibrium generalisation is given as }
\begin{eqnarray} \label{free_energyN}
F(\hat \rho) = E(\hat \rho) - \frac{1}{\beta} S(\hat \rho),
\end{eqnarray}
where $E(\hat \rho)$ is the average energy, $S(\hat \rho)$ is von Neumann entropy, and $\beta$ is the inverse temperature of the bath.  
In multiple of physical scenarios, free energy properly describes the directionality of state transitions if the system is in contact with a thermal bath. In particular, without any additional external driving, the system always thermalizes by decreasing its free energy, with the potential achieving its minimum at equilibrium. Moreover, if the system is able to perform work, the free energy dictates its maximal value, i.e., 
\begin{eqnarray} \label{second_law}
W \le - \Delta F.
\end{eqnarray} 
This powerful relation, known from classical thermodynamics, was also derived for driven systems that are weakly \cite{Aberg2013} or strongly coupled \cite{Seifert2016, Strasberg2019} to the thermal bath, as well as within the framework of the so-called thermal operations, where an explicit work storage is introduced \cite{Brandao2015,Skrzypczyk2014, BiswasFDR2022}. Remarkably, the inequality \eqref{second_law} can also be formulated as the (Jarzynski) equality if the free energy is treated as the fluctuating quantity \cite{Jarzynski1997, Esposito2009, Campisi2011, Alhambra2016}.  

\pmaz{On the other hand, if the system is isolated from a thermal bath, or the bath itself is finite, the inequality \eqref{second_law} cannot be saturated for processes leading to maximal change in free energy (reached in the Gibbs state), unless the initial state is Gibbs.} Considering an isolated system driven by a cyclic force, Allahverdyan et al. \cite{Allahverdyan2004} have replaced the free energy with a new thermodynamic quantity: ergotropy, defined as
\begin{eqnarray} \label{Defn:ergotropy}
R(\hat \rho) = E(\hat \rho) - \min_{\hat U} E(\hat U \hat \rho \hat U^\dag):= E(\hat{\rho})-E(\hat{\rho}_P),
\end{eqnarray}
where the minimum is taken over all unitary operators $\hat U$. If the state $\hat{\rho}$ is diagonal one can restrict to minimization over the set of permutations \cite{Ruch1976}. The minimum is achieved for the state $\hat{\rho}_P$ referred to as the passive state of the system, and $E(\hat\rho_P)$ is called the passive energy of the state $\hat{\rho}$. As it was shown in the seminal paper \cite{Allahverdyan2004}, ergotropy can be directly related to the free energy (cf. Section \ref{Chap2} below).  

The concept of the ergotropy appears in many areas of quantum thermodynamics. Firstly, it characterizes the optimal energy extracted from the so-called quantum batteries (see, e.g., \cite{Alicki2013, Binder2015, Campaioli2017, Monsel2020, Barra2020, Mir_passive_charging}). 
Recently, it was related to measures of entanglement \cite{Mir_Biparty, Mir_multiparty, Mir_independence}, coherences \cite{Francica2020} and relative entropy \cite{Sone2021}. It is worth noticing that ergotropy is in general not additive. This leads to the concept of complete passivity: the state is completely passive if the state composed of arbitrary number of its copies is passive. Remarkably, this links the idea of ergotropy with the equilibrium state, since any completely passive state has necessarily the Gibbs form \cite{Pusz1978}. Moreover, the complete passivity was recently related to the inability of extracting ergotropy in a presence of a catalyst \cite{Sparaciari2017}.

Ergotropy is also connected with an explicit model of the translationally-invariant energy-storage, called the ideal weight. The ideal weight model was introduced in \cite{Skrzypczyk2014} to formulate a thermodynamic framework which satisfies the Second Law. In particular, the authors proved that work $W$ performed via lifting the weight (increasing its average energy) is always smaller than the non-equilibrium free energy of the quantum state $\hat \rho$ (remaining in contact with a thermal reservoir): 
\begin{eqnarray} \label{skrzypczyk_work_extraction}
   W \le F(\hat \rho) - F(\hat \gamma_{\beta}) = \frac{1}{\beta}S(\hat \rho \| \hat \gamma_{\beta}),
\end{eqnarray}
where $\hat \gamma_{\beta}$ is the Gibbs state at inverse temperature of the bath $\beta$ and relative entropy $S(\hat \rho \| \hat \gamma_{\beta})$ is given by
\begin{equation}\label{eq:Defn_Relative_Entropy}
S(\hat \rho \| \hat \gamma_{\beta}) = \Tr\big(\hat{\rho}(\log \hat{\rho} -\log \hat\gamma)\big).
\end{equation}
 This result was subsequently generalized to an arbitrary correlated initial state of the weight and the system, and the extracted work was expressed as  \cite{Lobejko2020}:
\begin{eqnarray} \label{work_control_marginal}
    W = - \Delta R = R(\hat \sigma_i) - R(\hat \sigma_f),
\end{eqnarray}
where $\Delta R$ is a change in ergotropy of the so-called effective control-marginal state of the system $\hat \sigma_{i,f}$ (with indices $i$ and $f$ standing for initial and final states, respectively), containing information about coherences and correlations influencing the work extraction process \cite{Lobejko2020}. This reveals the connection between ergotropy and translational symmetry of the ideal weight model. Eq. \eqref{work_control_marginal} was further used in a formulation of the ``tight Second Law'', involving the ultimate bound given in Eq. \eqref{skrzypczyk_work_extraction}, but lowered by corrections expressing the effect of initial coherences and finite-size of the heat bath \cite{Lobejko2021}.

Since ergotropy can be understood as work done by the system, it is worth to explore its relevance in the operation of quantum thermal machines. 
Quantum thermal machines were characterized from multiple perspectives: starting from quantum optics \cite{Scovil1959, Scully2002,Jacobs2012,Goold2016}, thermalisation hypothesis in many-body systems \cite{Wilming2016,Perarnau2018}, dynamical description in quantum open systems \cite{Alicki1979}, and, most recently, theory of information \cite{ Rio2011} and resource theories \cite{Horodecki2003,Horodecki2013,Aberg2014,Brandao2015,Ng2015}. From the theoretical perspective, the family of heat engines in quantum regime can be categorized into two fundamental classes: autonomous and minimal coupling engines. Autonomous engines are the ones where the working body is simultaneously coupled to both heat baths, as well as to a  system storing energy  \cite{BrunnerVirtualqubit2012,PopescuSmallestHeatEngine2010}. On the other hand, minimal coupling engines act in discrete strokes \cite{Lobejko2020}. A simple construction of minimal coupling heat engines would enable us to emphasise and investigate the role of ergotropy extraction.  On the contrary, for autonomous engines, the difficulty in optimizing work production and efficiency lies in the necessity for taking into account all possible simultaneous interactions between working body and heat baths, thus restricting the analysis only to particular physical models \cite{JMosnel2018,Nimmrichter2017,KosloffLevy2014,Niedenzu2019,Lindenfels2019,SinghPRR2020}.


In this paper, we address the problem of ergotropy extraction via an energy conserving process. \pmaz{While recent interest in storing energy in quantum systems (\textit{quantum batteries}) focuses mostly on cases when ergotropy is extracted via the time-dependent driving  \cite{Andolina2018, Andolina2019}, here we concentrate solely on a fundamental, thermal contact scenario.} In other words, we ask how much ergotropy can be extracted via a heat flow through the system coupled to the thermal bath.  Since the ergotropy quantifies the amount of inversion of the population, one should expect that this is only possible via the non-Markovian evolution. Consequently, we use thermal operation framework \cite{Janzing2000,streater2009statistical, Horodecki2013} to optimize over such evolutions, emerging from energy conserving interactions. Our goal is to establish a connection between the ergotropy extraction and the Second Law of thermodynamics, as well as to characterize quantitatively the process of ergotropy extraction for finite and infinite dimensional systems. \pmaz{Alternatively, the problem of battery charging was recently formulated in terms of repeated interactions with a subsystem of  environment \cite{Barra_2019}, where a dissipative process leading to an active steady state of the battery was identified}. Secondly, we focus on the application of ergotropy extraction process in a design of thermal machines. We introduce a simple model of quantum heat engines under a name of open-cycle engines. Open-cycle engines form a subclass of the aforementioned minimal coupling engines. The simple operational idea of the open-cycle engine is to use the gradient of temperatures to extract ergotropy from the hot bath using the state taken from the cold reservoir. The name open-cycle is motivated by the demand to discard or thermalize the used resource after the energy is stored in the work reservoir. The complete analytical analysis of these machines in terms of efficiency and power production per cycle is performed for two- and three-dimensional working body systems. 

Here are the most important results that we obtain:
\begin{itemize}
    \item We formulate the Second Law of thermodynamics for a task of ergotropy extraction via any complete positive trace preserving (CPTP) map which preserves the Gibbs state of the system. The formulation is through an upper bound on the value of extracted ergotropy, expressed solely in terms of the non-equilibrium free energy; 
    \item We construct the process which saturates the bound for the quantum harmonic oscillator in the zero-temperature limit. 
    \item  
    On the opposite side of the spectrum of the dimensionality of the system, we identify optimal operations for ergotropy extraction for qubits and qutrits. We obtain this desrciption within the toy model of open-cycle engines, in which ergotropy extraction powers work production.
    \item We prove that optimal work production and efficiency of open-cycle engines is achieved for extremal thermal processes, disregarding the dimension of the working body.  
\end{itemize}

This paper is organized as follows. We start by introducing ergotropy of an isolated system and a system coupled with a bath in Section \ref{Chap2}. We derive upper bounds on ergotropy in terms free energy and relative entropy, and discuss the cases when these bounds can be saturated. In Section \ref{CHap3} we define the amount of ergotropy that can be extracted via any CPTP map that preserves the Gibbs state of the system and derive the upper bound in terms of quantum relative entropy (non-equilibrium free energy difference).  We prove that the upper bound can be saturated for a harmonic oscillator system initialized in the ground state. We also provide the numerical simulation for systems with finite number of levels which shows how the value of ergotropy extraction tends to the general bounds. We proceed to describe the open cycle heat engines in Section \ref{Chap4}, where we characterize a set of finite number of states over which optimal work production and efficiency can be achieved. Finally, we provide a complete analysis of the work production and efficiency for an open cycle heat engine with a qubit and qutrit working body. We conclude with the discussion in Section \ref{sec:5}.

\section{Ergotropy and free energy}\label{Chap2}

We start from discussing the established bounds on ergotropy of an isolated system \cite{Alicki2013} and system coupled with a bath \cite{Lobejko2021} and their saturation. 

From the definition of free energy and ergotropy given in Eq. \eqref{free_energyN} and Eq. \eqref{Defn:ergotropy}, and using the fact that entropy is invariant under unitary evolution,  we immediately have the following relation
\begin{eqnarray}
 R(\hat\rho) &=& E(\hat\rho)-E(\hat{\rho}_P) =  E(\hat\rho)-\frac{1}{\beta}S(\hat\rho)+\frac{1}{\beta}S(\hat\rho_P)-E(\hat{\rho}_P)\nonumber\\
 &=& F(\hat\rho)-F(\hat\rho_P)\label{eq:Ergotropy3}=\frac{1}{\beta}\big(S(\hat{\rho}\|\hat{\gamma}_{\beta})-S(\hat{\rho}_P\|\hat{\gamma}_{\beta})\big),
\end{eqnarray}
where $\hat{\rho}_P$ is the passive state of the system as given in Eq. \eqref{Defn:ergotropy}, and relative entropy is defined in Eq. \eqref{eq:Defn_Relative_Entropy}.
Note that the above holds for $\beta$ taking arbitrary values.

Among states with fixed entropy, energy of a system is minimized for the Gibbs state. This can be used in Eq. \eqref{eq:Ergotropy3} to minimize the free energy component: 
\begin{equation}\label{eq:Bound1}
    R(\hat{\rho}) = F(\hat{\rho})-F(\hat{\rho}_P) \leq F(\hat \rho) - F(\hat \gamma_{\beta^*}) = \frac{1}{\beta^{*}}S(\hat \rho \| \hat \gamma_{\beta^*}),
\end{equation}
where the inverse temperature $\beta^{*}$ is now uniquely specified by the relation $S(\hat{\rho})=S(\hat \gamma_{\beta^*})$. We see that the bound cannot be achieved in a single shot scenario if spectrum of $\hat{\rho}$ is different than that of $\hat\gamma_{\beta^*}$, as in this case it is impossible to achieve $\gamma_{\beta^{*}}$ from $\hat\rho_P$ using an unitary transformation. However, Alicki and Fannes showed in \cite{Alicki2013} that this bound can be achieved by processing asymptotically many number of copies of the state $\hat\rho$, i.e.,
\begin{equation}\label{Alicki}
    \lim_{N\rightarrow\infty}\frac{1}{N}R(\hat{\rho}^{\otimes N}) = F(\hat \rho) - F(\hat \gamma_{\beta^*}) = \frac{1}{\beta^{*}}S(\hat \rho \| \hat \gamma_{\beta^*}).
\end{equation}\\

We proceed by reporting a bound on the ergotropy of a system coupled with a bath at inverse temperature $\beta$ \cite{Lobejko2021}. Consider the system in the state $\hat{\rho}$ coupled with a thermal bath in Gibbs state $\hat{\tau}_{\beta}$ at inverse temperature $\beta$. The composite state of the system and bath is expressed as a product $\hat{\rho}\otimes\hat{\tau}_{\beta}$. Based on Eq. \eqref{eq:Ergotropy3}, one can provide the following upper bound on the ergotropy for a system in contact with a bath:
 \begin{equation}\label{eq:Bound2}
\begin{split}
    R(\hat \rho \otimes \hat{\tau}_{\beta}) &= F(\hat \rho \otimes \hat{\tau}_{\beta}) - F((\hat \rho \otimes \hat{\tau}_{\beta})_P) \\
    &= F(\hat \rho \otimes \hat{\tau}_{\beta})  - F( \hat{\gamma}_{\beta}\otimes\hat \tau_{\beta} ) + F( \hat{\gamma}_{\beta}\otimes\hat \tau_{\beta} ) - F((\hat \rho \otimes \hat{\tau}_{\beta})_P)\\
    &= \frac{1}{\beta} S(\hat \rho \| \hat \gamma_{\beta}) - \frac{1}{\beta} S((\hat \rho \otimes \hat{\tau}_{\beta})_P \| \hat{\gamma}_{\beta}\otimes\hat \tau_{\beta})\\&\leq \frac{1}{\beta} S(\hat \rho \| \hat \gamma_{\beta})= F(\hat{\rho})-F(\hat \gamma_{\beta}),
\end{split}
\end{equation}
where $\hat{\gamma}_{\beta}$ is the Gibbs state of the system in the inverse temperature of the bath $\beta$. This is similar to the bound given in Eq. \eqref{eq:Bound1}, with the crucial difference that the inverse temperature $\beta$ is now fixed and inherited from the bath. In \cite{Skrzypczyk2014,Lobejko2021}, the saturation of the bound is established for a quasi-reversible process consisting of infinitely many steps conserving the average energy, under the assumption of presence of an infinitely large heat bath.  


\section{Free energy bound for ergotropy extraction}\label{CHap3}

In the previous section, we have discussed bounds on ergotropy of an isolated system and a system coupled with a bath in equilibrium. These bounds and their saturation relate the ergotropy of the system with the non-equilibrium free energy (see Eq. \eqref{eq:Bound1} and Eq. \eqref{eq:Bound2}). This suggest that the ergotropy can be understood as a generalization of the free energy for finite-dimensional quantum systems \cite{Lobejko2021}. 

In this section, we move to the question of how much ergotropy can be extracted through system's coupling with a heat bath. We define ergotropy extraction  for a state $\rho$ evolving via a CPTP map $\Phi$ as the change of ergotropy in the process, i.e.:
\begin{equation}
    R(\Phi(\hat\rho))-R(\hat\rho).
\end{equation}
 We are interested in inducing ergotropy on a system through an interaction with a single heat bath, and we assume the total energy of the system and bath to be conserved. As a consequence, we consider CPTP maps $\Phi$ of the form:
\begin{equation}\label{eq:DefnTO}
    \Phi(\hat{\rho})=\Tr_{B}(U(\hat{\rho}\otimes\hat\tau_{\beta})U^{\dagger}),
\end{equation}
such that $[U,H_S+H_B]=0$, where $H_S$ and $H_B$ are the Hamiltonian of the system and the bath, respectively. As before, $\hat\tau_{\beta}$ is the Gibbs state of the bath at inverse temperature $\beta$. These maps are known as thermal operations \cite{Horodecki2013, Janzing2000, streater2009statistical}.
\pmaz{Let us stress that the defining property (\ref{eq:DefnTO}) is not restrictive, and serves only as a formulation of the first law of thermodynamics in the spirit of collision theory. Namely, one should consider the above unitaries $U$ as generated by an arbitrary interaction Hamiltonian, with the interaction negligible at first and last stages of the process, when the systems remain spatially separated.}
Thermal operations constitute a subclass of a bigger set of transformations called Gibbs-preserving maps, which preserve the Gibbs state of the system:  $\Phi(\hat\gamma_{\beta})=\hat\gamma_{\beta}$. The sets of Gibbs preserving maps and thermal operations are equal when restricted to states $\hat\rho$ which are diagonal in the energy eigenbasis i.e., $[\hat\rho,H_S]=0$.

Before addressing the question of ergotropy extraction via thermal operations, we propose a bound on the extraction by a general CPTP map.
\begin{prop}
 Consider an arbitrary state $\hat{\rho}$ that is evolving via a CPTP map $\Phi$. Then, the ergotropy of the final state $\Phi(\hat\rho)$ can be decomposed into:
 \begin{eqnarray} \label{ergotropy_relative_entropy}
 R(\Phi(\hat\rho)) = \frac{1}{\beta} S(\hat \rho\|\hat{\gamma}_{\beta}) - \frac{1}{\beta} S(\Phi(\hat \rho)_P\|\hat{\gamma}_{\beta}) - \frac{1}{\beta} \left( S(\hat \rho\|\hat{\gamma}_{\beta})-S(\Phi(\hat \rho)\|\hat{\gamma}_{\beta})\right),
 \end{eqnarray}
 where $\hat \gamma_{\beta}$ is a Gibbs state of the system at some inverse temperature $\beta$. 
\end{prop}
The formula \eqref{ergotropy_relative_entropy} follows from the application of the identities given in Eq. \eqref{eq:Ergotropy3} :
\begin{eqnarray}
R(\Phi(\hat\rho)) = F(\Phi(\hat\rho))-F(\Phi(\hat\rho)_P), \ \ \frac{1}{\beta} S(\hat \rho \| \hat \gamma_{\beta}) = F(\hat \rho) - F(\hat \gamma_{\beta}).
\end{eqnarray}
In Eq. \eqref{ergotropy_relative_entropy} we identify three different contributions to the ergotropy of the final state $\Phi(\hat\rho)$. The first one $\frac{1}{\beta} S(\hat \rho\|\hat{\gamma}_{\beta})$ is the free energy difference of the initial state $\hat\rho$ and the Gibbs state $\hat\gamma_{\beta}$. This expression has appeared in the derivations of the bounds on ergotropy (see Eq. \eqref{eq:Bound1} and \eqref{eq:Bound2}), and provides the ultimate bound for the work extraction from a single heat bath (see Eq. \eqref{skrzypczyk_work_extraction}). The second contribution $-\frac{1}{\beta} S(\Phi(\hat \rho)_P\|\hat{\gamma}_{\beta})$ is always non-positive (due to the non-negativity of the relative entropy) and characterizes a thermodynamic distance of the final passive state to the equilibrium. In particular, this term is zero if and only if the final state $\Phi(\hat \rho)$ has the same set of eigenvalues as the Gibbs state, so that $\Phi(\hat \rho)_P = \hat \gamma_{\beta}$. Finally, the last term:
\begin{eqnarray}
\frac{1}{\beta} \left(S(\Phi(\hat \rho)\|\hat{\gamma}_{\beta}) - S(\hat \rho\|\hat{\gamma}_{\beta}) \right) = F(\Phi(\hat \rho)) - F(\hat \rho),
\end{eqnarray}
can be interpreted as the entropy production for a CPTP map $\Phi$. Indeed, for the maps $\Phi$ being thermal operations as given by (\ref{eq:DefnTO}), the change of the system energy can be identified with heat, i.e., $\Delta E = Q$, and then $F(\Phi(\hat \rho)) - F(\hat \rho) = Q - \frac{1}{\beta} \Delta S$. Considering the map $\Phi$ to be Gibbs-preserving leads to the crucial corollary
\begin{cor}If the CPTP map $\Phi$ is Gibbs-preserving, i.e., $\Phi(\hat \gamma_{\beta}) = \hat \gamma_{\beta}$, then ergotropy of the final state $\Phi(\hat\rho)$ is upper-bounded as
\begin{eqnarray}\label{second_law_ergotropy}
R(\Phi(\hat\rho)) \le \frac{1}{\beta} S(\hat \rho\|\hat{\gamma}_{\beta}).
\end{eqnarray}
\end{cor}
\begin{proof}
The proof of the corollary immediately follows from two properties of the relative entropy and Gibbs-preserving map: the non-negativity of the relative entropy and monotonicity of the relative entropy under Gibbs preserving maps: 
\begin{eqnarray}
\frac{1}{\beta} S(\Phi(\hat \rho)_P\|\hat{\gamma}_{\beta}) &\geq& 0, \label{eq:Nonnegativity}\\
\frac{1}{\beta} \left( S(\hat \rho\|\hat{\gamma}_{\beta})-S(\Phi(\hat \rho)\|\hat{\gamma}_{\beta})\right) &\geq& 0, \label{eq:Monotonicity}
\end{eqnarray}
Inserting these Eq. \eqref{eq:Nonnegativity} and Eq. \eqref{eq:Monotonicity} in Eq. \eqref{ergotropy_relative_entropy}, we obtain the bound.
\end{proof}
 The inequality \eqref{second_law_ergotropy} shows that the final ergotropy cannot be higher than the initial thermodynamic resource given by the non-equilibrium free energy. This remains in agreement with the bound (\ref{skrzypczyk_work_extraction}), derived for a model which defined work through changes of the avarage energy of the quantum weight. As mentioned earlier, in case of no quantum correlations between the ideal weight and the system, its ergotropy can be associated with maximal change of average energy of the weight. Then the bound \eqref{second_law_ergotropy} is just another statement of the Second Law. 

This result can be generalised to encapsulate ergotropy extraction for non-passive initial states via Gibbs-preserving maps.  
\begin{thm}
[Ultimate bound on extraction of ergotropy]The change of the ergotropy of the state $\hat\rho$ evolving via a Gibbs preserving map $\Phi$ 
(i.e. satisfying $\Phi(\hat{\gamma}_{\beta})=\hat{\gamma}_{\beta}$) is upper bounded as
\begin{equation}\label{eq:BOEforGPTransformationthm}
    R(\Phi(\hat\rho))-R(\rho) \leq \frac{1}{\beta}S(\hat\rho_P\|\hat{\gamma}_{\beta}).
\end{equation}
\end{thm}
\begin{proof}
We begin by writing the change in ergotropy using Eq. \eqref{eq:Ergotropy3} as 
\begin{eqnarray}
R(\Phi(\rho))-R(\rho)
                    &=& \Big(F(\Phi(\hat{\rho}))-F(\Phi(\hat\rho)_P)\Big)-\Big(F(\hat{\rho})-F(\hat{\rho}_P)\Big)\nonumber\\&=& \Big(F(\hat{\rho}_P)-F(\Phi(\hat{\rho})_P)\Big)-\Big(F(\rho)-F(\Phi(\hat\rho))\Big)\nonumber\\
                    &\leq&  \Big(F(\hat\rho_P)-F(\Phi(\hat\rho)_P)\Big) = \Big(F(\hat\rho_P)-F(\hat\gamma_{\beta})+F(\hat\gamma_{\beta})-F(\Phi(\hat\rho)_P)\Big)\nonumber\\
                 &=&\frac{1}{\beta}\Big(S(\hat\rho_P\|\hat\gamma_{\beta})-S(\Phi(\hat\rho)_P\|\hat\gamma_{\beta})\Big) \leq \frac{1}{\beta}S(\hat\rho_P\|\hat\gamma_{\beta}),
\end{eqnarray}
where the first inequality follows from the monotonicity of free energy under the Gibbs-preserving map  and the second inequality follows from non-negativity of relative entropy.
\end{proof}  
Note that the bound given in Eq. \eqref{eq:BOEforGPTransformationthm} does not depend on a particular  CPTP Gibbs-preserving map $\Phi$, but only on the initial state $\hat\rho$. This bound is consistent with the previous bound in Eq. \eqref{second_law_ergotropy}, since $S(\hat\rho_P\|\hat\gamma_{\beta})\leq S(\hat\rho\|\hat\gamma_{\beta})$. 
If the initial state $\hat\rho$ is already passive $R(\hat\rho)=0$, then inequality in Eq. \eqref{eq:BOEforGPTransformationthm} boils down to Eq. \eqref{second_law_ergotropy}. Therefore, we can interpret Eq. \eqref{eq:BOEforGPTransformationthm} as the second law for ergotropy extraction that sets the upper bound on the amount of ergotropy that can be extracted via a Gibbs preserving transformation. 

In the next section, we design an example showing saturation of the bound for $\Phi$ given by a thermal operation \eqref{eq:DefnTO}. 

\subsection{The maximal-energy thermal process and saturation of the bound on ergotropy extraction}\label{saturation}

\pmaz{In order to identify a scenario in which the general bound (\ref{eq:BOEforGPTransformationthm}) is saturated, we restrict ourselves to a subset of Gibbs preserving maps, generated by thermal operations. We recall the notion of thermal processes \cite{Horodecki2013}, which completely characterise the action of a thermal operation on the diagonal of a quantum state.}
\begin{lemma}\label{HO_IMP_LEM}
 Consider a given state $\hat\rho$ that is diagonal in the energy basis, i.e.,
 \begin{equation}\label{eq:Diagonal state}
     \hat{\rho}=\sum_{i=1}^d \langle \epsilon_i|\hat{\rho}|\epsilon_i\rangle|\epsilon_i\rangle\langle \epsilon_i|,
 \end{equation}
where $\{|\epsilon\rangle\}_{i=1}^d$ are energy eigenvectors, and a thermal operation $\Phi$ which transforms the state $\hat{\rho}$ to $\hat{\sigma}$. This is equivalent to 
existence of a stochastic matrix $A_\Phi$ such that 
\begin{eqnarray}\label{GPcondition}
A_\Phi\v{p}&=&\v{q},\\
A_\Phi\v{\gamma}_{\beta}&=&\v{\gamma}_{\beta},
\end{eqnarray}
where $\v{p}$, $\v{q}$ and $\v{\gamma}_{\beta}$ are probability vectors generated from diagonal entries of $\hat{\rho}$, $\hat{\sigma}$, and $\hat{\gamma}_{\beta}$ respectively, i.e.,
 \begin{eqnarray}
 \v{p} &=& \Big(\langle \epsilon_1|\hat{\rho}|\epsilon_1\rangle  \ldots \langle \epsilon_d|\hat{\rho}|\epsilon_d\rangle \Big),\nonumber \\
 \v{q} &=&  \Big(\langle \epsilon_1|\hat{\sigma}|\epsilon_1\rangle  \ldots \langle \epsilon_d|\tilde{\sigma}|\epsilon_d\rangle \Big),\nonumber \\
 \v{\gamma}_{\beta} &=& \Big(\frac{e^{-\beta\epsilon_1}}{Z}  \ldots \frac{e^{-\beta\epsilon_d}}{Z}\Big).\nonumber
 \end{eqnarray}
The stochastic matrix $A_{\Phi}$ associated with thermal operation $\Phi$ is called thermal process. 
\end{lemma}

The bound on ergotropy extraction given by \eqref{eq:BOEforGPTransformationthm} is saturated for a thermal operation $\Phi_{\text{max}}$ acting on a ground state of a harmonic oscillator system, $\hat H=\sum_{n} n\omega |n\rangle\langle n|$. Its associated  thermal process $A_{\Phi_{max}}$ is given by 

\begin{eqnarray}\label{map}
 A_{\Phi_{\text{max}}}= \begin{pmatrix}
W & B\\
\Omega & \v{0}
\end{pmatrix},
\end{eqnarray}
where 
\begin{eqnarray}
 W&:=&
\begin{pmatrix}
0 & 0 & 0 & 0 & \ldots &0\\
0 & 1 & 0 & 0 & \ldots &0\\
0 & 0 & 1 & 0 & \ldots &0\\
\vdots & \vdots & \vdots & \vdots & \vdots & \vdots\\
0 & 0 & 0 & 0 & \ldots &1
\end{pmatrix}_{(L-1)\times(L-1)}
\quad\quad B:=
\begin{pmatrix}
1 & 1 & 1 & 1 & \ldots & \ldots\\
0 & 0 & 0 & 0 & \ldots & \ldots\\
0 & 0 & 0 & 0 & \ldots & \ldots\\
\vdots & \vdots & \vdots & \vdots & \vdots & \vdots\\
0 & 0 & 0 & 0 & \ldots & \ldots
\end{pmatrix}_{(L-1)\times\infty},\\
\Omega&:=&\frac{1}{Z}\begin{pmatrix}
1 & 0 & 0 & 0 & \ldots &0\\
e^{-\beta\omega} & 0 & 0 & 0 & \ldots &0\\
e^{-2\beta\omega} & 0 & 0 & 0 & \ldots &0\\
e^{-3\beta\omega} & 0 & 0 & 0 & \ldots &0\\
\vdots  & \vdots & \vdots & \vdots & \ldots &\vdots\\
\end{pmatrix}_{\infty\times(L-1)}
\quad\quad\quad\quad \v{0}:=\begin{pmatrix}
0 & 0 & 0 & 0 & \ldots & \ldots\\
0 & 0 & 0 & 0 & \ldots & \ldots\\
0 & 0 & 0 & 0 & \ldots & \ldots\\
\vdots & \vdots & \vdots & \vdots & \vdots & \vdots\\
0 & 0 & 0 & 0 & \ldots & \ldots
\end{pmatrix}_{\infty\times\infty}
\end{eqnarray}
with the choice
\begin{equation}\label{L}
    L = 1+\frac{1}{\beta\omega}\log Z,
\end{equation}
taking integer values for  $\frac{1}{\beta\omega}\log Z \in \mathbb{N}$, where $Z=\Tr(e^{-\beta \hat H})$. One can immediately see  that $ A_{\Phi_{\text{max}}}$ is stochastic, and that it preserves the probability vector
\begin{equation}
    \v{\gamma}_{\beta} = \Big(\frac{1}{Z}, \quad \frac{e^{-\beta\omega}}{Z}, \quad \frac{e^{-2\beta\omega}}{Z} , \ldots \Big).
\end{equation}
 The action of $A_{\Phi_{\text{max}}}$ on the probability vector $(1 \; 0  \ldots 0)^{T}$ determines the  transformation of the ground state:
\begin{eqnarray}
\Phi_{\text{max}}\Big(|0\rangle\langle 0|\Big)=\hat G = \frac{1}{Z}\sum_{n=L}^{\infty}e^{-\beta\omega (n-L)}|n\rangle\langle n|.
\end{eqnarray}
It can be seen that the state $\hat G$ is the Gibbs state $\hat \gamma_\beta$ with populations shifted by $\omega L$. In order to calculate the ergotropy of the final state, we find the corresponding passive state $\hat G_P$, i.e., the minimal energy state obtained via the permutation of occupation probabilities. For this infinite-dimensional system, the passive state is found in the limit where all "zeros", which occupy the first $L$ energy levels, are "sent to infinity" (as in the Hilbert's Hotel paradox). In this limit, we have $\hat G_P=\hat\gamma_{\beta}$, i.e. the passive state of $\hat G$ is the Gibbs state. The proposed construction of the passive state does not apply in finite dimension, since the Gibbs state $\hat\gamma_{\beta}$ in this case does not have any "zeros" within its spectrum. 

Putting $G_P=\hat\gamma_{\beta}$ turns the LHS of the inequality \eqref{eq:BOEforGPTransformationthm} to
\begin{eqnarray}
R(\hat G)-R(|0\rangle\langle 0|) = E(\hat{G})-E(\hat{G}_P) = E(\hat{G})-E(\hat\gamma_{\beta}) = \frac{1}{\beta}\log Z,
\end{eqnarray}
while the RHS of the inequality is equal to:
\begin{eqnarray}\label{eq:ergotropychange}
\frac{1}{\beta}S(|0\rangle\langle 0|\|\hat \gamma_{\beta})=\frac{1}{\beta}\log Z.
\end{eqnarray}

\begin{figure}
    \centering
    \includegraphics[width = 0.6 \textwidth]{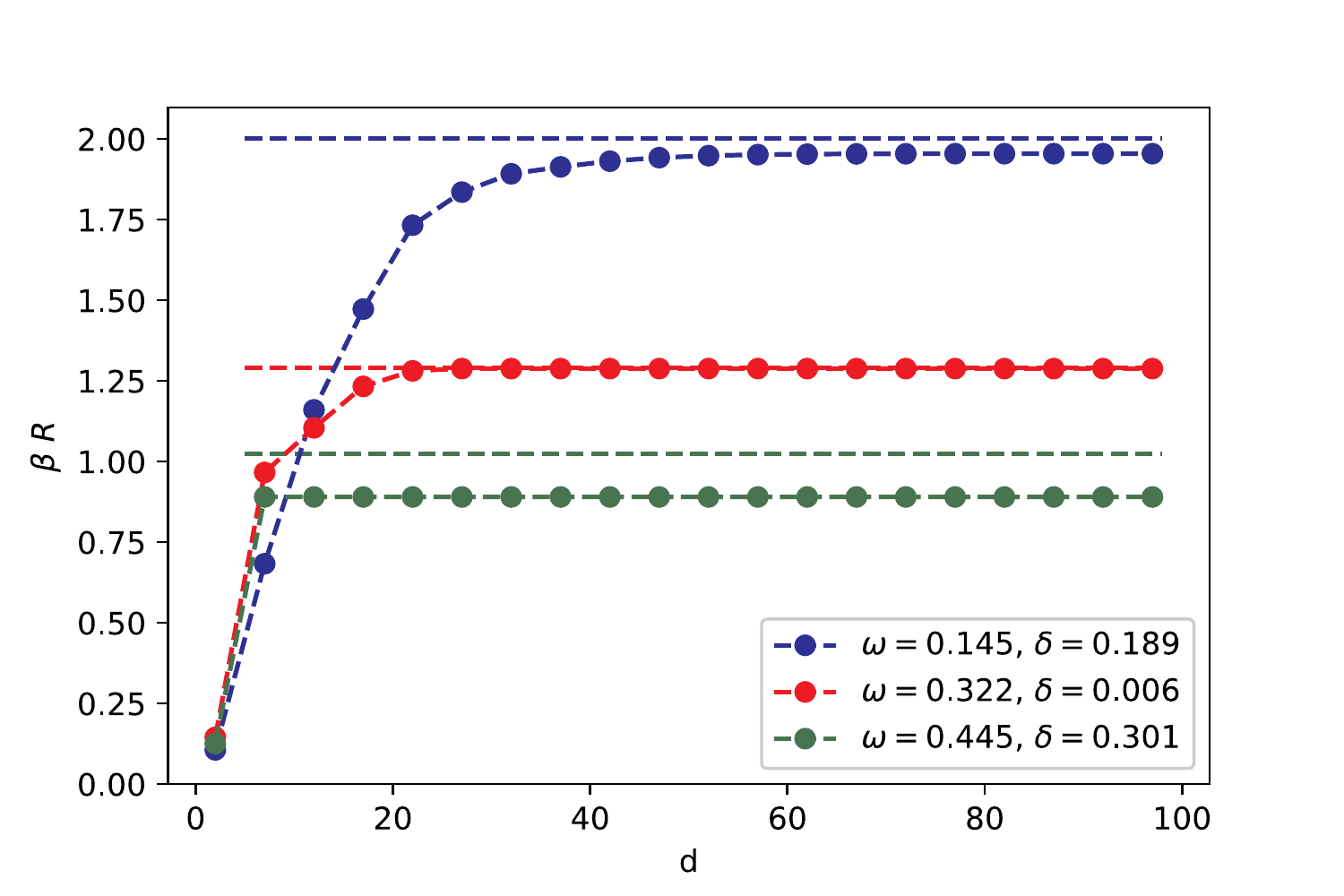}
    \caption{Ergotropy extraction for the $d$-dimensional qudit with frequency $\omega$ prepared in the ground state, for the thermal process at temperature $\beta$ that maximizes the final energy. Horizontal dashed lines correspond to the ergotropy extraction bound \eqref{eq:BOEforGPTransformationthm}. It can be seen that the final ergotropy for the maximal-energy process saturates with growing dimension $d$, and it approaches the optimal ergotropy if  $\delta \to 0$.}
    \label{fig:saturation}
\end{figure}

We see that the harmonic oscillator in the zero-temperature limit (i.e., in the ground state) and with frequency $\omega$ satisfying $\frac{1}{\beta\omega}\log Z \in \mathbb{N}$ saturates the bound of the ergotropy extraction  \eqref{eq:BOEforGPTransformationthm}. The thermal process which lead to the saturation maximizes the final energy of the state. This is visible from the fact that all the action of the process acting on the ground state $|0\rangle\langle 0|$ is described by its first column. It can be easily verified that for the $n$ row of the first column of $\Omega$, the value of the corresponding matrix element is upper bounded by $e^{-\beta\omega (L+n-2)}$ (otherwise the map is not Gibbs-preserving). However, from (\ref{L}) we read the value of this element to be equal to it: $\frac{e^{-\beta\omega(n-1)}}{Z}=e^{-\beta\omega(L+n-2)}$. Therefore, the map leads to maximal possible occupation of high energy states, at the expense of the lowest $L-1$ energy levels.

Below, we extend to a finite dimension the protocol maximazing energy of the final state, and calculate the amount of extracted ergotropy. Our numerical calculations are based on the fact that the state with maximal energy which can be obtained via a thermal operation, can be associated with a so-called thermomajorozation curve which is tightly-thermomajorized by the curve associated with the initial state, and which has the so-called descending $\beta$-order. We refer the reader to Lemma \ref{energy_max} of Supplementary Material for the detailed description of the protocol. 

Let us define the detuning parameter:
\begin{eqnarray}
\delta = \min \left[\frac{1}{\beta\omega}\log Z - \lfloor \frac{1}{\beta\omega}\log Z \rfloor, \lceil \frac{1}{\beta\omega}\log Z \rceil - \frac{1}{\beta\omega}\log Z \right].
\end{eqnarray}
In Fig. \ref{fig:saturation} we present the results of the numerical simulation for the final ergotropy of the $d$-dimensional system with different frequencies $\omega$ and non-zero value of the detuning parameter $\delta$, with the ergotropy extraction protocol maximizing the energy of the final state. In the regime $\delta\rightarrow 0$, extractable ergotropy approaches the bound already for moderate dimension $d$, while the non-zero detuning parameter prohibits the strict saturation even in the limit $d\rightarrow\infty$. Note that in the limit $\beta\rightarrow 0$, thermal processes become permutations, and the protocol maximizing energy of the final state is the one maximizing its ergotropy. Therefore, the numerical results presented in Fig. \ref{fig:saturation} indicate that no thermal operation is able to strictly saturate the bound \eqref{eq:BOEforGPTransformationthm} for non-zero detuning $\delta$ and high temperatures of the bath, with harmonic oscillator prepared in the ground state.

In the next section, we move to the case of small dimensional systems, in which the bound for ergotropy extraction can be far from being saturated. We analyze these cases in the context of utilizing small dimensional system as working bodies in thermal machines which extract ergotropy from the environment. We will show that for finite temperatures of the hot bath, optimal protocols for ergotropy extraction dependents on the spectrum of the Hamiltonian and bath temperatures, even if we relax cyclicity constraints by demanding that the initial state is in thermal equilibrium with the cold bath.

\section{The open-cycle heat engines}\label{Chap4}

In this section, we consider the model of the open-cycle heat engine and describe the role of ergotropy extraction from the bath as work that can be stored in a battery modelled by ideal weight \cite{Skrzypczyk2014}. We shall analyze the optimal performance of such an engine and characterize the states for which optimal work production and efficiency can be achieved. 
We consider open cycle engine within the family of minimal coupling engines  \cite{Lobejko2021}, where  the working body alternately couples to the heat baths and the battery, resembling traditional stroke engines. 

\subsection{Description of the open cycle heat engine}
An open cycle heat engine is a minimal coupling thermal machine introduced in \cite{Lobejko2020}, consisting of two baths in equilibrium associated with two different inverse temperatures $\beta_H$ and $\beta_C$, and a battery. We denote by $\beta_H$ the inverse temperature of the hot heat bath, while  $\beta_C$ stands for the inverse temperature of the cold bath, and $\beta_{H}<\beta_{C}$. The working body is initially in a thermal state with respect to the cold bath. In subsequent stages, it alternately couples to the hot bath and the battery through discrete two-body energy conserving operations referred to as strokes.  \\
\begin{figure}[t]
\includegraphics[width=11cm]{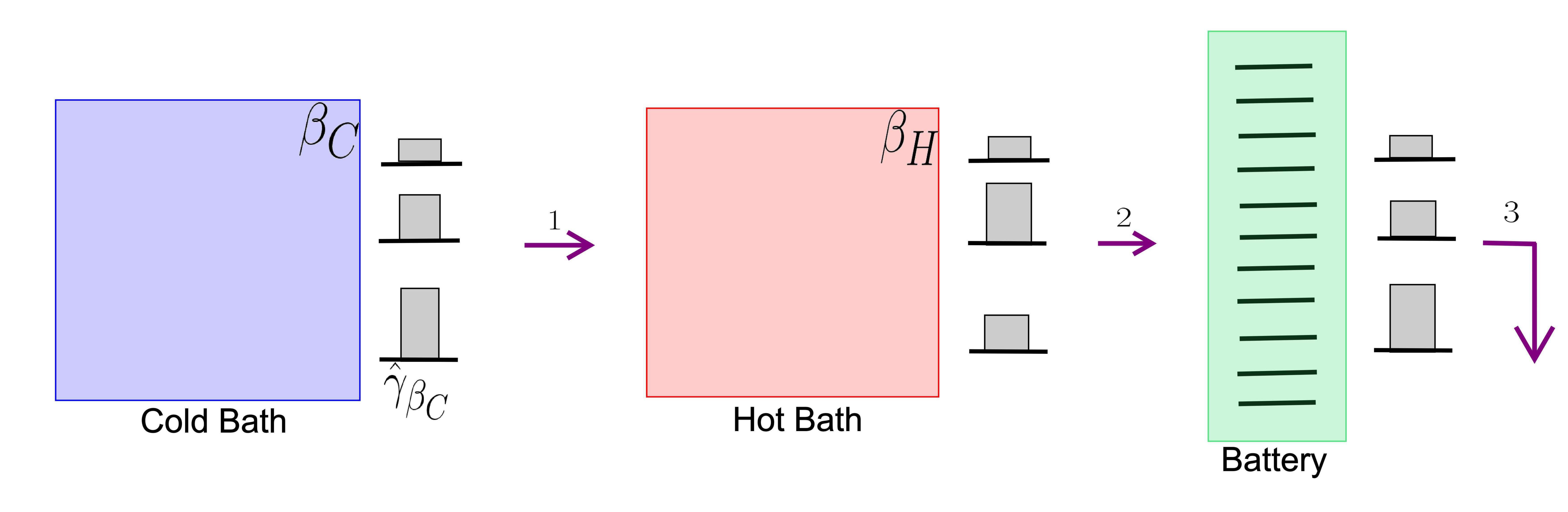}
\centering
\caption{\textbf{A graphical representation of open-cycle heat engine with a qutrit working body:} This engine is composed of a cold and hot bath at inverse temperature $\beta_C$ and $\beta_H$ respectively, and a battery as a work storage. Initially, the working medium is thermalized with respect to cold bath. In stroke 1, working body extract ergotropy from hot bath, that induces  non-passivity in the working body. In stroke 2, ergotropy of the working body is stored in the battery as work. This interaction makes the state of the working body passive. Finally, in stroke 3, working body can be discarded or can be thermalized with the cold bath.}
\label{fig:Open_Cycle_qutrit_Engine}
\end{figure}

The design of the open-cycle engine is schematically presented in Fig.~\ref{fig:Open_Cycle_qutrit_Engine}. As in the minimal coupling quantum heat engines, open-cycle engine in its first stroke utilizes the interaction between the hot bath and the working body to extract ergotropy. In the second stroke, extracted ergotropy is stored as work on the battery. In the final stroke, the working body gets thermalized with the cold bath (or discarded). We proceed with the detailed description of the engine.

\subsubsection{Working body thermalized with respect to cold bath.}\label{cold}

At the beginning of the procedure, we initialize both of the baths in equilibrium i.e.,
\begin{equation}\label{Gibbs} 
    \hat \tau_{\beta_H} = \frac{e^{-\beta_H\hat {H}_H}}{\Tr(e^{-\beta_H\hat {H}_H})} \quad\quad \hat \tau_{\beta_C} = \frac{e^{-\beta_C\hat {H}_C}}{\Tr(e^{-\beta_C\hat {H}_C})} 
\end{equation}
where $\hat{H}_H$ and $\hat{H}_C$ are Hamiltonians of hot and cold reservoir, respectively. We assume that heat capacities of the baths are infinite, thus the temperatures baths remain unchanged through the operation of the engine. The working body is a $d$-dimensional system initially in equilibrium with cold bath. 
The Hamiltonian associated with the working body is given by 
\begin{equation}
    \hat{H}_{S}=\sum_{i=0}^{d}\omega_{i}|\epsilon_i\rangle\langle \epsilon_i|,
\end{equation}
while state can be written as 
\begin{equation}\label{state_of_WB}
    \hat{\rho}_{S}=\frac{e^{-\beta_{C}\hat{H}_{S}}}{\Tr(e^{-\beta_{C}\hat{H}_{S}})}.
\end{equation}  
The initial state of the engine can be written as 
\begin{equation} \label{initial_state}
    \hat \rho = \hat \rho_{S} \otimes \hat \tau_{\beta_{H}} \otimes \hat\rho_{B}, 
\end{equation}
where $\rho_{B}$ is a state of the battery. The free Hamiltonian $\hat{H}_0$ of the combined system is given by 
\begin{eqnarray}
\hat{H}_0 =\hat{H}_S+\hat{H}_H+\hat{H}_B,
\end{eqnarray}
where $\hat{H}_S$, $\hat{H}_H$ , $\hat{H}_B$ are Hamiltonian of the system, hot bath, and battery respectively. If the open-cycle engine runs iteratively, such that after storing the ergotropy of the working body in the battery it has been thermalized again with cold bath, then  the engine at the beginning of each of the runs is affected only by the change of the battery state  $\rho_{B}$. The role of the working medium is to steer the energy flow from hot bath to the battery, as described below.

\subsubsection{Interaction of working body with the hot heat bath and battery via discrete two body energy conserving strokes.}\label{inter}

Following the definition of stroke operations in \cite{Lobejko2020}, we allow only for the  interactions among any two bodies at a given moment. This implies that the unitary evolution of an engine can be decomposed into a product of two unitaries: 
\begin{equation} \label{total_unitary}
    \hat U = \hat U_{B} \hat{U}_{H},
\end{equation}
where $\hat{U}_{H}$ describes the interaction between the hot bath and the working body, $\hat{U}_{B}$ describes the interaction between the system and the battery. 

In order to account for energy in the engine, we demand that the unitaries conserve the total energy. Therefore, we impose the constraint 
\begin{equation}\label{comut}
[\hat{U}_{H},\hat{H}_{S}+\hat{H}_{H}]=0 \;\quad \text{and} \quad\; [\hat U_{B},\hat{H}_{S}+\hat{H}_{B}]=0.
\end{equation}

The purpose of a unitary  $\hat{U}_{H}$ is to extract ergotropy from the hot bath to the working body (ergotropy extraction step) in the first stroke, while $\hat U_{B}$ aims at charging the battery by converting the ergotropy that has been extracted from the hot bath into work in the second stroke. This justifies the ordering of the unitaries given in Eq. \eqref{total_unitary}.

The interaction between the hot bath and the working body generates $\hat{U}_{H}$, which induces a thermal operation $\mathcal{E}_{H}$ given in Eq. (\ref{eq:DefnTO}) which leads to extraction of ergotropy from the hot bath:
\begin{equation}\label{mapTO}
    \hat{\rho}_{S}\rightarrow \Tr_{H}[\hat{U}_{H} (\hat{\rho}_{S}\otimes\hat{\tau}_{H})\hat{U}_{H}^{\dagger}]:=\mathcal{E}_{H}(\hat{\rho}_S).\end{equation}

As the initial state $\hat{\rho}_S$ of the working body is in equilibrium with cold bath, we have $R(\hat{\rho}_S)=0$. Therefore, for a work storage modeled by the ideal weight (see discussion in Supplementary Material \ref{SubWorkStotage}), work stored at it in the second stroke is equal to ergotropy extracted in the first stroke \cite{Skrzypczyk2014,Lobejko2020}:
\begin{equation}\label{eq:Defn_of_Work}
    W(\mathcal{E}_H(\hat{\rho}_S))= R(\mathcal{E}_H(\hat{\rho}_S).
\end{equation}

Since the working body interacts with the hot bath, the heat flow from the hot bath to the working body increases its average energy. The transferred amount of heat $Q$ is defined as the change of average energy of the system due to interaction with the hot bath, i.e., 
\begin{equation}\label{Defn:HeatExchanged}
    Q(\mathcal{E}_{H}(\hat{\rho}_S))= \Tr(\hat{H}_S\big(\mathcal{E}_{H}(\hat{\rho}_S)-\hat{\rho}_S\big)).
\end{equation}
In agreement with Clausius's statement of second law of thermodynamics, in the absence of external work source, heat flows from hot to cold body. The definition (\ref{Defn:HeatExchanged}) abides this, which ensures that the transferred amount of heat is always non-negative:
\begin{lemma}
[Directionality of heat flow from hot bath to a working body that is thermalized with cold bath]
\label{Postive_Heat}
Consider a state $\hat\rho_S$ that is  initially thermalized with cold bath at temperature $\beta_C$, transforming via any thermal operation $\mathcal{E}_H$ with respect to the hot bath at temperature $\beta_H$, then the amount of the exchanged heat, given in Eq. (\ref{Defn:HeatExchanged}), is always non-negative i.e.,
\begin{equation}
    Q(\mathcal{E}_{H}(\hat{\rho}_S))= \Tr(\hat{H}_S\big(\mathcal{E}_{H}(\hat{\rho}_S)-\hat{\rho}_S\big)) \geq 0.
\end{equation}
\end{lemma}

\begin{proof}[Proof sketch:]
From the linearity of the function $Q(\mathcal{E}_{H}(\hat{\rho}_S))$ we conclude that minimal value of the exchanged amount of heat is obtained at one of the extremal points of the set of states $\mathcal{T}(\hat{\rho}_S)$ - the set of states which can be achieved by applying a thermal operations on $\hat{\rho}_S$. One of those extremal points 
is the state $\hat{\rho}_S$ itself. 
Thus we have to show that minimum of the average energy of extremal points of $\mathcal{T}(\hat{\rho}_S)$  is never smaller than average energy of $\hat{\rho}_S$. 

To this end we shall first note that there is only one extremal element of this set 
with $\beta$-order (for definition of $\beta$-order see section \ref{MPofTO} of appendix) given by $(1,2,\ldots, d)$. Next we show, that for any state with different  $\beta$-order, there exists a thermal operation $\mathcal{A}$ which decreases its average energy. This implies that for extremal states of $\mathcal{T}(\hat{\rho}_S)$ which are different than $\hat{\rho}_S$, their average energy is strictly higher than average energy of some other states from $\mathcal{T}(\hat{\rho}_S)$
(i.e. those resulting from action of $A$). Therefore the minimum must be achieved on $\hat{\rho}_S$, which gives the minimal value 
of $Q$ equal to zero. 
The complete proof can be found in section \ref{PROOF_positive_heat_flow} of the Supplementary Material.
\end{proof}

Finally, we define the efficiency of open cycle heat engine as the ratio work stored in the battery with the amount of heat exchanged from the hot bath i.e.,
\begin{equation}\label{eq:Defn_efficiency}
    \eta(\mathcal{E}_H(\hat{\rho}_S)) = \frac{W(\mathcal{E}_H(\hat{\rho}_S)}{Q(\mathcal{E}_H(\hat{\rho}_S)}=\frac{R(\mathcal{E}_H(\hat{\rho}_S)}{Q(\mathcal{E}_H(\hat{\rho}_S)}.
\end{equation}
In the next section we shall discuss in detail the optimal performance of the open cycle heat engine. Precisely, how one can deposit the maximum amount of work in the battery and perform it the most efficiently? This question is equivalent to extraction of maximum ergotropy from the hot heat bath as one can see from Eq. \eqref{eq:Defn_of_Work}.

\subsection{Optimal performance for the open cycle heat engine} \label{optimal_performance_section}
Below, we optimize the work and efficiency defined in Eq. \eqref{eq:Defn_of_Work} and Eq. \eqref{eq:Defn_efficiency} over the set of states that can be achieved from the state of working body $\hat{\rho}_S$ via a thermal operation. As a set of thermal processes forms a polytope, it immediately follows from lemma \ref{HO_IMP_LEM} that the set of states achievable via thermal operations from any given diagonal state also forms a polytope. We denote by $\mathcal{T}(\hat{\rho}_S)$ the set of states which can be  achieved from $\hat{\rho}_S$ via a thermal operation. This set is called thermal polytope of the state $\hat{\rho}_S$. 
Therefore, we are interested in maximizing work and efficiency over all the states in $\mathcal{T}(\hat{\rho}_S)$ i.e.,
\begin{equation}\label{eq:Max_Erg_and_Max_eff}
    W_{\text{max}} = \max_{\hat\sigma\in\mathcal{T}(\hat{\rho}_S)} R(\hat\sigma), \quad \quad \eta_{\max} = \max_{\hat\sigma\in\mathcal{T}(\hat{\rho}_S)} \eta(\hat\sigma).
\end{equation}
 The theorem below characterises the states in the $\mathcal{T}(\hat{\rho}_S)$ which lead to optimal work and efficiency.    

\begin{thm}[Optimal performance of open cycle engine]\label{Thm_of_extremal_pts}
 The maximum value of work and efficiency for an open cycle quantum heat engine achieved at extremal points of the set of states $\mathcal{T}(\hat{\rho}_S)$.
\end{thm}

\begin{proof} Optimal work is obtained by maximizing the ergotropy function over the polytope $\mathcal{T}(\hat{\rho}_S)$ (see Eq. \eqref{eq:Max_Erg_and_Max_eff}). From the definition of ergotropy given in Eq. \eqref{Defn:ergotropy}, we can prove the convexity for ergotropy as follows,
\begin{eqnarray}
R(\alpha \hat\rho+(1-\alpha)\hat\sigma)&=&\max_{U\in \;\mathcal{U}(d)}\Tr \Big(H\big(\alpha \hat\rho+(1-\alpha) \hat\sigma-U(\alpha \hat\rho+(1-\alpha) \hat\sigma)U^{\dagger}\big)\Big)\nonumber\\
&=& \max_{U\in \;\mathcal{U}(d)}\Tr\Big(\big(\alpha H\big(\hat\rho-U\hat\rho U^{\dagger}\Big)+(1-\alpha)H\big(\hat\sigma- U\hat\sigma U^{\dagger}\big)\Big)
\nonumber\\
&\leq& \alpha\max_{U\in \;\mathcal{U}(d)} \Tr\big( H(\hat\rho- U\hat\rho U^{\dagger})\big)+(1-\alpha)\max_{U\in \;\mathcal{U}(d)}\Tr\big( H(\hat\sigma- U\hat\sigma U^{\dagger})\big)\nonumber\\
&=&\alpha R(\hat\rho)+(1-\alpha)R(\hat\sigma),
\end{eqnarray}
where $\hat\rho$ and $\hat\sigma$ are any arbitary states and $\alpha\in [0,1]$. We exploit the fact that if we maximize a convex function over a polytope, the function takes maximal values at the extremal points. This implies that the maximum value of work production is achieved at an extremal point of $\mathcal{T}(\hat{\rho}_S)$. 

For maximizing efficiency, we notice that an arbitrary state $\hat\sigma\in\mathcal{T}(\hat\rho_S)$ can be written as $\hat\sigma=\sum_i\lambda_i\hat\sigma_i$, where $\{\hat\sigma_i\}$ are extremal points of $\mathcal{T}(\hat{\rho}_S)$ and $\lambda_i\in[0,1]$ for each $i$. Thus,  we can write from Eq. \eqref{eq:Defn_of_Work} and \eqref{eq:Defn_efficiency} the following 
\begin{eqnarray}
\eta(\hat\sigma) = \frac{R(\hat\sigma)}{Q(\hat\sigma)} = \frac{R(\sum_{i}\lambda_i\hat\sigma_i)}{Q(\sum_{i}\lambda_i\hat\sigma_i)} &\leq& \frac{\sum_i\lambda_i R(\hat\sigma_i)}{\sum_i\lambda_i Q(\hat\sigma_i)}\nonumber\\&=&\sum_j\Big(\frac{\lambda_j Q(\hat{\sigma}_j)}{\sum_i\lambda_i Q(\hat\sigma_i)}\Big) \frac{R(\hat{\sigma}_j)}{Q(\hat{\sigma}_j)}\leq\max_j \frac{R(\hat{\sigma}_j)}{Q(\hat{\sigma}_j)} = \max_j \eta(\hat\sigma_j),
\end{eqnarray}
where the first inequality exploits convexity of ergotropy and linearity of the function $Q(\cdot)$, and the second inequality uses lemma \ref{Postive_Heat} ($Q(\hat{\sigma}_j)$ is non-negative for all extremal states, and therefore $\Big(\frac{\lambda_j Q(\hat{\sigma}_j)}{\sum_i\lambda_i Q(\hat\sigma_i)}\Big) \in [0,1]$. 
\end{proof}

 Therefore, maximization of work and  efficiency of the open-cycle heat engine is tantamount to optimisation over all extremal points of $\mathcal{T}(\hat{\rho}_S)$. Note that these two may not be achieved simultaneously. Now we shall discuss the open cycle heat engine with qubit and qutrit working body and analyze their optimal performance, in this way addresing the ergotropy extraction process for low dimensional systems. We use the results of \cite{Mazurek_2018} to characterize the extremal points of the set $\mathcal{T}(\hat{\rho}_S)$.

\subsubsection{Open cycle heat engine with qubit working body}
Consider a qubit working body $\hat{\rho}_S$ with Hamiltonian
\begin{equation}
     \omega |1\rangle\langle 1|
\end{equation}
in thermal state at inverse temperature $\beta_C$  given by 
\begin{eqnarray}\label{eq:WBQubit}
\hat{\rho}_S=\frac{1}{1+e^{-\beta_C\omega}}\begin{pmatrix}
1 & 0\\
0 & e^{-\beta_C\omega}
\end{pmatrix}.
\end{eqnarray}
 From \cite{Horodecki2013}, it is straightforward to characterize $\mathcal{T}(\hat{\rho}_S)$ as convex hull with two extremal points as follows
\begin{equation}\label{eq:Convhull2}
    \mathcal{T}(\hat{\rho}_S) = \text{conv }\Bigg\{\hat{\rho}_S \;,\;\; \frac{1}{1+e^{-\beta_C\omega}}\begin{pmatrix}
1+(e^{-\beta_C\omega}-e^{-\beta_H\omega}) & 0\\
0 & e^{-\beta_H\omega}
\end{pmatrix}\Bigg\},
\end{equation}
where $\hat{\rho}_S$ is given in Eq. \eqref{eq:WBQubit}. From Theorem \ref{Thm_of_extremal_pts}, we optimise the performance of open cycle heat engine over extremal points of $\mathcal{T}(\hat{\rho}_S)$ as given in Eq. \eqref{eq:Convhull2}. Therefore, we evaluate work production  and efficiency of the engine at
\begin{equation}
    \frac{1}{1+e^{-\beta_C\omega}}\begin{pmatrix}
1+(e^{-\beta_C\omega}-e^{-\beta_H\omega}) & 0\\
0 & e^{-\beta_H\omega}
\end{pmatrix},
\end{equation}
which is the only extremal point that gives non-trivial work production and efficiency. We calculate the optimal work output defined in Eq. \eqref{eq:Defn_of_Work} as 
\begin{equation}\label{W}
    W = \omega\Big(\frac{2e^{-\beta_H\omega}}{1+e^{-\beta_C\omega}}-1\Big).
\end{equation}
Therefore, the engine produces non-zero work iff $2e^{-\beta_H\omega}-1>e^{-\beta_C\omega}$. To calculate the efficiency, we proceed by calculating the amount of heat exchanged given in Eq. \eqref{Defn:HeatExchanged} as 
\begin{equation}
    Q= \frac{\omega}{1+e^{-\beta_C\omega}}(e^{-\beta_H\omega}-e^{-\beta_C\omega}).
\end{equation}
Therefore, the optimal efficiency is given by
\begin{equation}
    \eta =  1-\frac{(1-e^{-\beta_H\omega})}{(e^{-\beta_H\omega}-e^{-\beta_C\omega})}.
\end{equation}
Let us now compare with the more general case of minimal coupling quantum heat engines \cite{Lobejko2020}, for which we have
\begin{eqnarray}
    W_{\text{mc}} &=& \omega\Big(\frac{2e^{-\beta_H\omega}}{1+e^{-(\beta_C+\beta_H)\omega}}-1\Big),\label{Wgen}\\
    \eta_{\text{mc}} &=& 1-\frac{(1-e^{-\beta_H\omega})}{\big(e^{-\beta_H\omega}-e^{-(\beta_C+\beta_H)\omega}\big)}.
\end{eqnarray}
Thus $W_{\text{mc}}\geq W$ and $\eta_{\text{mc}}\geq\eta$, with the equality holding in the limit $\beta_C\rightarrow \infty$. This relation is in agreement with the fact that open cycle engines are minimal coupling engines where working body goes to its initial state (after storing work) by  thermalization with the cold bath. From
(\ref{W}) and (\ref{Wgen}) it is also visible that the range of temperatures in which the machine can work (non-zero work is produced) is larger for minimal coupling heat engines. Now we move to the qutrit case of open-cycle heat engines, which, to the best of our knowledge, for minimal coupling heat engines remains uncharacterized.

\subsubsection{Open cycle heat engine with a qutrit working body}
As a working body we select a qutrit working body with Hamiltonian 
\begin{equation}
    \omega_1|1\rangle\langle 1|+ \omega_2|2\rangle\langle 2|
\end{equation}
that is initialized in thermal state given by
\begin{equation}\label{eq:WBQutrit}
    \hat{\rho}_S = \frac{1}{1+e^{-\beta_C\omega_1}+e^{-\beta_C\omega_2}}\begin{pmatrix}
1 & 0 & 0\\
0 & e^{-\beta_C\omega_1} & 0\\
0 & 0 & e^{-\beta_C\omega_2}
\end{pmatrix}.
\end{equation}
 Like before, by employing Theorem \ref{Thm_of_extremal_pts} we optimise work output and efficiency over all extremal points of $\mathcal{T}(\hat{\rho}_S)$. To do so, we employ a result from \cite{Mazurek_2018} which characterizes the set $\mathcal{T}(\hat{\rho}_S)$ as a convex hull of its extremal points. We are going to state the lemma by using the following notation
\begin{eqnarray}\label{eq:Notation}
q^H_{ij} &=& e^{-\beta_H(\omega_i-\omega_j)},\\
q^C_{ij} &=& e^{-\beta_C(\omega_i-\omega_j)},
\end{eqnarray}
where $\omega_0=0$. We rewrite $\hat{\rho}_S$ given in Eq. \eqref{eq:WBQutrit} as
\begin{equation}\label{eq:WBQutrit2}
    \hat{\rho}_S=\frac{1}{\mathcal{Z}_C}\Big(|0\rangle\langle 0|+q^C_{10}|1\rangle\langle 1|+q^C_{20}|2\rangle\langle 2|\Big),
\end{equation}
where $\mathcal{Z}_C= 1+ q^C_{10}+q^C_{20}$. The application of \cite{Mazurek_2018} (see Supplementary Material, Lemma \ref{extremal} for details) results in the following:
\begin{lemma}
[Extremal points of the qutrit thermal polytope] \label{QTC} The set of achievable states via thermal operation in presence of bath at inverse temperature $\beta_H$, from the initial state $\hat{\rho}_S$ given in Eq. \eqref{eq:WBQutrit2}, is given by the convex hull of the extremal points:
 \begin{eqnarray}
     \mathcal{T}(\hat{\rho}_S)= \begin{cases}\text{ conv }\{\hat{\rho}_S\;,\; \hat{\rho}^{1}_S \;,\; \hat{\rho}^{2}_S \;,\;\hat{\rho}^{3}_S \;,\; \hat{\rho}^{4}_S  \} &\quad \mathrm{for~} \beta_H\geq\beta_0 ,\\\text{ conv }\{\hat{\rho}_S\;,\; \hat{\rho}^{1}_S \;,\; \hat{\rho}^{2}_S \;,\;\hat{\rho}^{3}_S \;,\; \hat{\rho}^{5}_S \;,\; \hat{\rho}^{6}_S \} &\quad \mathrm{for~} \beta_H<\beta_0 ,
\end{cases}
\end{eqnarray}
 where $\beta_0$ is specified by the relation 
 \begin{eqnarray}
 e^{-\beta_0\omega_1}+e^{-\beta_0\omega_2}=1,
 \end{eqnarray}
 and extremal points are defined as
 \begin{eqnarray}
\hat{\rho}^{1}_S &=& \frac{1}{\mathcal{Z}_C}\Big((1-q^{H}_{10}+q^C_{10})|0\rangle\langle 0|+q^H_{10}|1\rangle\langle 1|+q^C_{20}|2\rangle\langle 2|\Big),\label{eq:1}\\
\hat{\rho}^{2}_S &=& \frac{1}{\mathcal{Z}_C}\Big((1|0\rangle\langle 0|+((1-q^H_{21})q^C_{10}+q^C_{20})|1\rangle\langle 1|+q^H_{21}q^C_{10}|2\rangle\langle 2|\Big),\label{eq:2}\\ 
\hat{\rho}^{3}_S &=& \frac{1}{\mathcal{Z}_C}\Big((1-q^{H}_{20}+q^{H}_{21}q^{C}_{10})|0\rangle\langle 0|+((1-q^H_{21})q^C_{10}+q^C_{20})|1\rangle\langle 1|+q^H_{20}|2\rangle\langle 2|\Big),\label{eq:3}\\
\hat{\rho}^{4}_S &=& \frac{1}{\mathcal{Z}_C}\Big((1+q^{C}_{10}+q^{C}_{20}-q^{H}_{10}-q^{H}_{20})|0\rangle\langle 0|+q^H_{10}|1\rangle\langle 1|+q^H_{20}|2\rangle\langle 2|\Big),\label{eq:4}\\
\hat{\rho}^{5}_S &=& \frac{1}{\mathcal{Z}_C}\Big(((q^{H}_{01}-q^{H}_{21})q^{C}_{10}+q^{C}_{20})|0\rangle\langle 0|+q^H_{10}|1\rangle\langle 1|+(1-q^{H}_{10}+(1-q^{H}_{01}+q^{H}_{21})q^C_{10})|2\rangle\langle 2|\Big),\label{eq:5}\\
\hat{\rho}^{6}_S &=& \frac{1}{\mathcal{Z}_C}\Big((q^{C}_{10}(q^H_{01}-q^H_{21})+q^{C}_{20})|0\rangle\langle 0|+(1-q^H_{20}+q^C_{10}(1-q^H_{01}+q^H_{21}))|1\rangle\langle 1|+q^H_{20}|2\rangle\langle 2|\Big),\label{eq:6}
 \end{eqnarray}
where $\mathcal{Z}_C= 1+ q^C_{10}+q^C_{20}$.
\end{lemma}
The maximization over these extremal points turns out to be difficult as it dependents on the values of $\beta_H$, $\omega_1$ and $\omega_2$. We start from evaluating work produced by the engine as the amount of ergotropy extracted from the hot bath, according to \eqref{eq:Defn_of_Work}.  We enlist the values in Table \ref{tab:1}. Values of heat exchanged with the hot bath for each of the extremal states are given in Table \ref{tab:2}, and  efficiency simply follows from taking the ratio of the amount of ergotropy extracted with corresponding amount of heat exchanged.

\begin{table}[]
    \begin{center}
    \begin{tabular}{|c|c|}
    \hline
    Extremal states & Amount of ergotropy extracted   \\
    \hline
    $\hat{\rho}^{1}_S$ & $\max\Big\{0,\frac{\omega_1}{Z_C}\Big(2q^{H}_{10}-1-q^C_{10}\Big)\Big\}$\\
    \hline
    $\hat{\rho}^{2}_S$ & $\max\Big\{0,\frac{1}{Z_C}(\omega_2-\omega_1)\Big(2q^{H}_{21}q^{C}_{10}-(q^{C}_{10}+q^{C}_{20}\Big)\Big\}$\\       
    \hline
    $\hat{\rho}^{3}_S$ & $\max\Big\{0, \frac{(\omega_{2}-\omega_{1})}{Z_C}\Big(q^{H}_{20}-q^{C}_{20}-(1-q^{H}_{21})q^{C}_{10}\Big),\;\frac{\omega_{2}}{Z_C}q^{H}_{20}-\frac{\omega_{1}}{Z_C}\Big((1-q^{H}_{20})+q^{H}_{21}q^{C}_{10}\Big)$\\& $-\frac{(\omega_{2}-\omega_{1})}{Z_C}\Big((1-q^{H}_{21})q^{C}_{10}+q^{C}_{20}\Big)\Big\}$\\
    \hline
     $\hat{\rho}^{4}_S$ & $\max\Big\{0,\;\frac{\omega_{1}}{Z_C}\Big(2q^{H}_{10}+q^{H}_{20}-(1+q^{C}_{10}+q^{C}_{20})\Big),\; \frac{\omega_1}{Z_C}(q^{H}_{10}-q^{H}_{20})+\frac{\omega_{2}}{Z_C}\Big(2q^{H}_{20}+q^{H}_{10}$\\
    & $-(1+q^{C}_{10}+q^{C}_{20})\Big)\Big\}$\\
    \hline
    $\hat{\rho}^{5}_S$ & $\max\Big\{0, \frac{\omega_1}{Z_C}\Big(q^{H}_{10}-(q^{H}_{01}-q^{H}_{21})q^{C}_{10}-q^{C}_{20})\Big),\; \frac{\omega_1}{Z_C}\Big(q^{H}_{10}-(1-q^{H}_{10})$\\
     &$-(1-q^{H}_{01}-q^{H}_{21})q^C_{10}\Big)+\frac{\omega_2}{Z_C}\Big((1-q^{H}_{10})+(1-q^{H}_{01}-q^{H}_{21})q^C_{10}$\\
     &$-(q^{H}_{01}-q^{H}_{21})q^{C}_{10}-q^{C}_{20}\Big)\Big\} $\\
    \hline
    $\hat{\rho}^{6}_S$ & $\max\Big\{0,\frac{\omega _2-\omega_1}{Z_C} \Big(2 q^H_{20}-1-q^C_{10} (1+q^H_{21}-q^H_{01})\Big),\;\frac{\omega _1}{{Z_C}} \Big(1-q^C_{20}-q^H_{20}$\\
    &$+q^C_{10}(1+2q^H_{21}-2q^H_{01})\Big),\;\frac{\omega_1}{Z_C}\Big(q^C_{10} (1+q^H_{21}-q^H_{01})+1-2q^H_{20}\Big)$ \\
    & $+\frac{\omega_2}{Z_C}\Big(q^H_{20}-q^C_{20}+q^C_{10}q^H_{21}-q^C_{10}q^H_{01}\Big),\frac{\omega_1}{Z_C}\Big(1 - q^C_{20} - q^H_{20}$\\
    &$+q^C_{10}(1 + 2 q^H_{21} - 2 q^H_{01})\Big) -\frac{\omega_2}{Z_C}\Big(1-2q^H_{20}$\\
   & $+q^C_{10} (1 + q^H_{21} - q^H_{01})\Big),$
\;$ \frac{\omega_2}{Z_C}\Big(q^H_{20}-q^C_{20}
   +q^C_{10}(q^H_{21}-q^H_{01})\Big)\Big\}$\\
    \hline
    \end{tabular}
\end{center}
\caption{Amount of ergotropy extracted for different extremal states, under different permutations.}
    \label{tab:1}
\end{table}

\begin{figure}[http]
\centering
\includegraphics[scale=0.7]{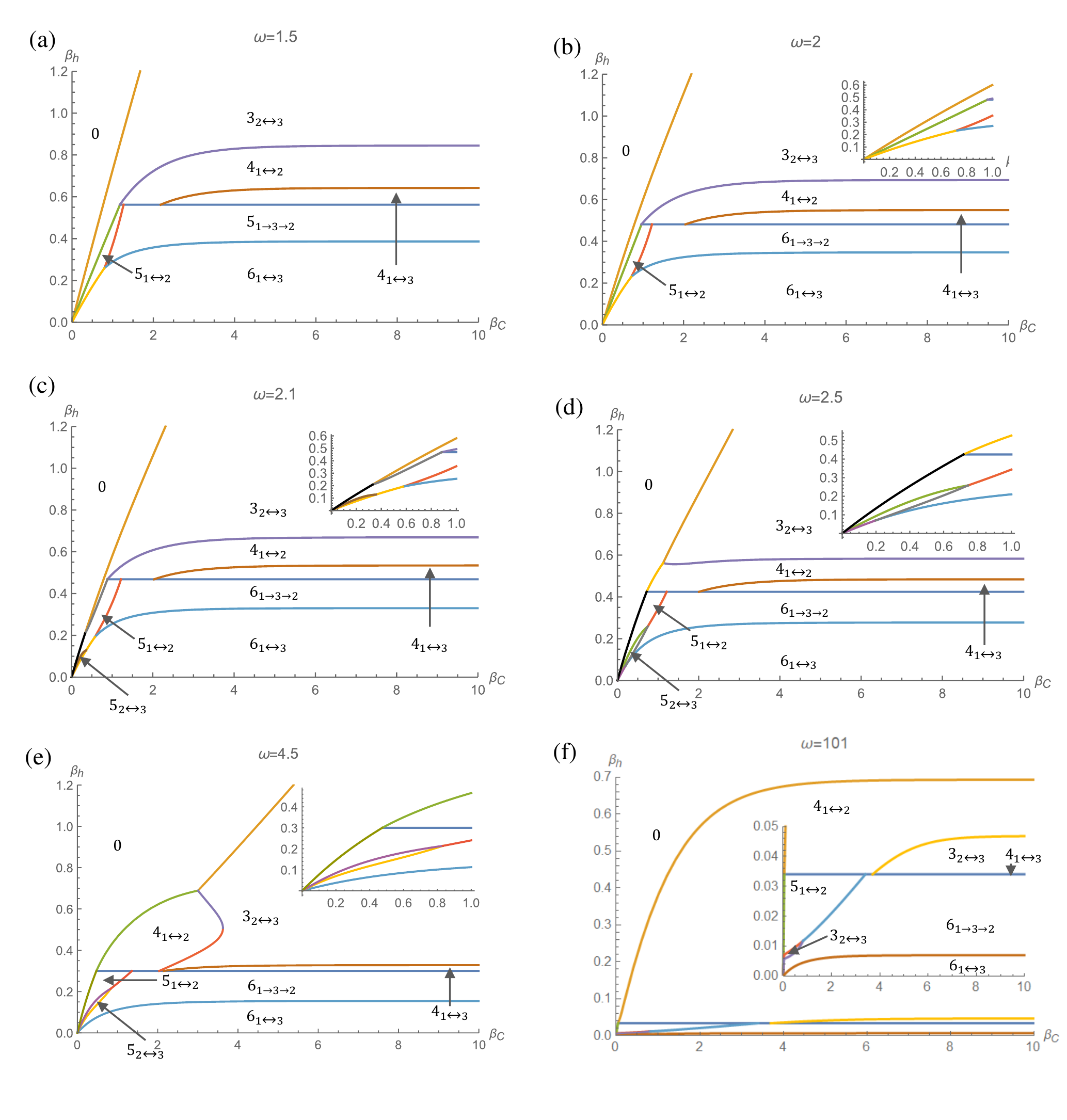}
\caption{\textbf{Protocols for optimal ergotropy extraction} from a  qutrit Gibbs state with  inverse temperature $\beta_{C}$, subjected to interaction with thermal bath with inverse temperature $\beta_{H}$. Hamiltonian has spectrum $(0, 1, \omega)$. (a) $\omega=1.5$
(b) $\omega=2$. (c) $\omega=2.1$. (d) $\omega=2.5$. (e) $\omega=4.5$. (f) $\omega=101$. 0 marks the region in which ergotropy extraction is not possible. Other regions are marked by symbols of the corresponding optimal final states (as described by eq. (\ref{eq:1})-(\ref{eq:6})). Subscripts describe the permutation leading to optimal ergotropy extraction. A transition between protocols $4$ and $5$ or $4$ and $6$ happens at $1=e^{-\beta_{H}}+e^{-\beta_{H}\omega}$.\label{rys}}
\end{figure}

\begin{figure}[http]
\centering
\includegraphics[scale=0.7]{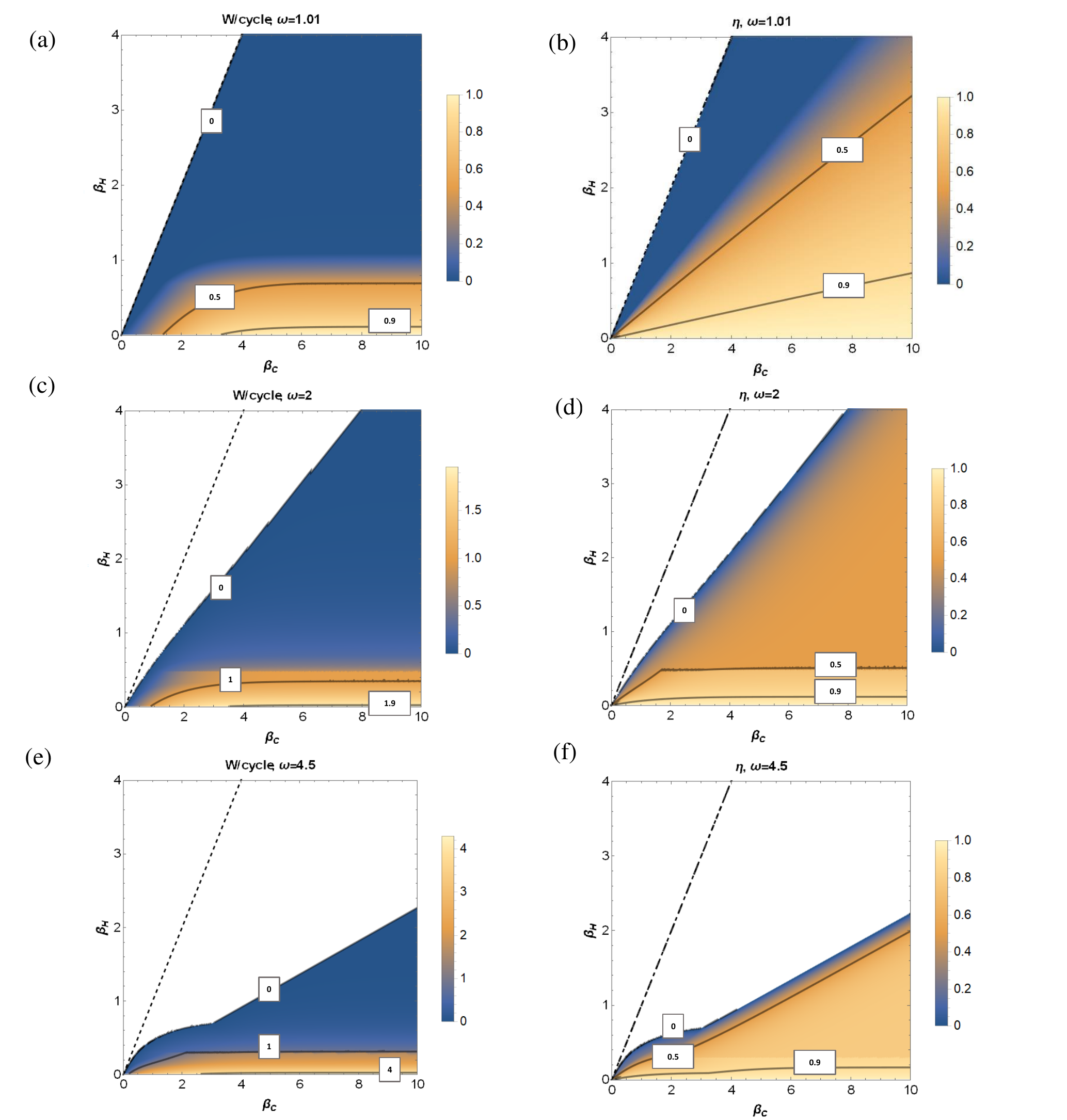}
\caption{Work production per cycle and efficiencies of the open-cycle engine for different qutrit spectrum: (a), (b): $\omega=1.01$ (c), (d): $\omega=2$; (e), (f): $\omega=4.5$. Dashed line constraints the operational region in the limit $\omega\rightarrow 1$. \label{Fig4}
}
\end{figure}

\begin{table}[]
\begin{center}
\begin{tabular}{ |c|c| } 
\hline
Extremal states & Amount of  Heat Exchanged  \\
\hline
$\hat{\rho}^{1}_S$ &  $\frac{1}{Z_C}\omega_1\Big(q^H_{10}-q^C_{10}\Big)$  \\
\hline
$\hat{\rho}^{2}_S$ &  $\frac{1}{Z_C}\Big(\omega_2-\omega_1\Big)\Big(q^C_{10}(q^H_{21}-q^C_{21})\Big)$\\
\hline
$\hat{\rho}^{3}_S$ &  $\frac{1}{Z_C}\Big(q^H_{20}\omega_2-q^C_{20}(\omega_2-\omega_1)-q^C_{10}q^H_{21}\omega_1)\Big)$\\
\hline
$\hat{\rho}^{4}_S$ &  $\frac{1}{Z_C}\Big((q^H_{10}-q^C_{10})\omega_1+(q^H_{20}-q^C_{20})\omega_2\Big)$\\
\hline
$\hat{\rho}^{5}_S$ &  $\frac{1}{Z_C}\Bigg(q^H_{10}(1-q^C_{10}q^H_{01})\omega_1+\Big(q^C_{10}q^H_{20}q^H_{01}-q^C_{20}+(1-q^C_{10}q^H_{01})(1-q^H_{10})\Big)\omega_2\Bigg)$\\
\hline
$\hat{\rho}^{6}_S$ & $\frac{1}{Z_C}\Big(1-q^C_{10}q^H_{01})(1-q^H_{20})\omega_1+(q^H_{20}-q^C_{20})\omega_2\Big)$\\
\hline
\end{tabular}
\caption{Amount of heat exchanged for different extremal states.}
    \label{tab:2}
\end{center}
\end{table}
As optimal values of ergotropy and efficiency dependent on the temperature regime and system Hamiltonian, we analyze the work production and efficiency for a qutrit Hamiltonian given by
\begin{equation}
    H=\dyad{1}+\omega\dyad{2}.
\end{equation}

We identify optimal protocols for work production. These can be characterized by different states of the working body (\ref{eq:1})-(\ref{eq:6})) obtained by application of a thermal operation to the initial state (see Fig. \ref{rys}). It should be stressed that, in contrast to the infinite dimensional case investigated before, for low dimensional systems the optimal protocols for ergotropy extraction may not maximize energy of the working body.   
In Fig. \ref{Fig4} we show the values of extractable ergotropy and efficiency of the engine. We see that the region in which the ergotropy is extracted shrinks with increasing $\omega$. In the limit $\omega\rightarrow\infty$ and $\beta_{C}\rightarrow\infty$ we recover the condition $\beta_{H}<\log{2}$, bounding the operational regime of a qubit minimal step engine \cite{Lobejko2020}. Indeed, for high values of $\omega$, extremal thermal processes acting on a qutrit Gibbs state cannot lead to higher population inversion than the ones acting on a qubit system (see Supplementary Material Sec. \ref{SubDet} eq. (\ref{A6}-\ref{A9})).

Maximization of energy (provided by a process leading to the state $\hat\rho^{6}_{S}$) is the optimal strategy in the limit $\beta_{H}\rightarrow 0$, but otherwise the landscape of optimal protocols 
is very complex. In particular, close to critical values of $\beta_{H}$ satisfying $1= e^{-\beta_{H}}+e^{-\beta_{H}\omega}$, efficiency and extracted work may strongly depend on $\beta_{H}$. This should be attributed to different structures of the set of extremal thermal processes above and below this limit \cite{Mazurek_2018}. We expect that diversity of optimal strategies increases with the dimension of the working body. 

\subsubsection{Open cycle heat engine with qudit working body}
As it was seen in the previous section, the structure of non-passive states becomes more complicated while going from the two-level to  three-level systems. Thus, in order to characterize the higher $d$-dimensional qudits, we compute optimal work production and efficiency via the numerical simulation. The method of computing the optimal quantities is based on the thermomajorization diagrams (which we discuss in Supplementary Material \ref{MPofTO}): For a particular initial state $\hat{\rho}_S$, we numerically search for all extremal points in the set $\mathcal{T}(\hat{\rho}_S)$.  

The results are presented in Fig. \ref{fig:qudit}. Due to the fact that number of extremal points in a thermal polytope grows with factorial of the dimension, we were not able to perform these calculations for moderate $d>8$.  
Nevertheless, optimizing over all possible protocols for ergotropy extraction suggests that work production and efficiency increase with the growing dimension of the working body, while the bound given in (\ref{eq:BOEforGPTransformationthm}) cannot be saturated for small dimensions. These conclusions are similar to those valid for a restricted class of protocols aiming at optimizing energy (Fig. \ref{fig:saturation}), and match the intuition that low dimensionality of the working body  enforces separation (in terms of relative entropy) between the Gibbs state and the state emerging from ergotropy extraction, prohibiting saturation of the bound. Therefore, using low dimensional working body may require careful optimization.   

\begin{figure}
    \centering
    \includegraphics[width = 0.45 \textwidth]{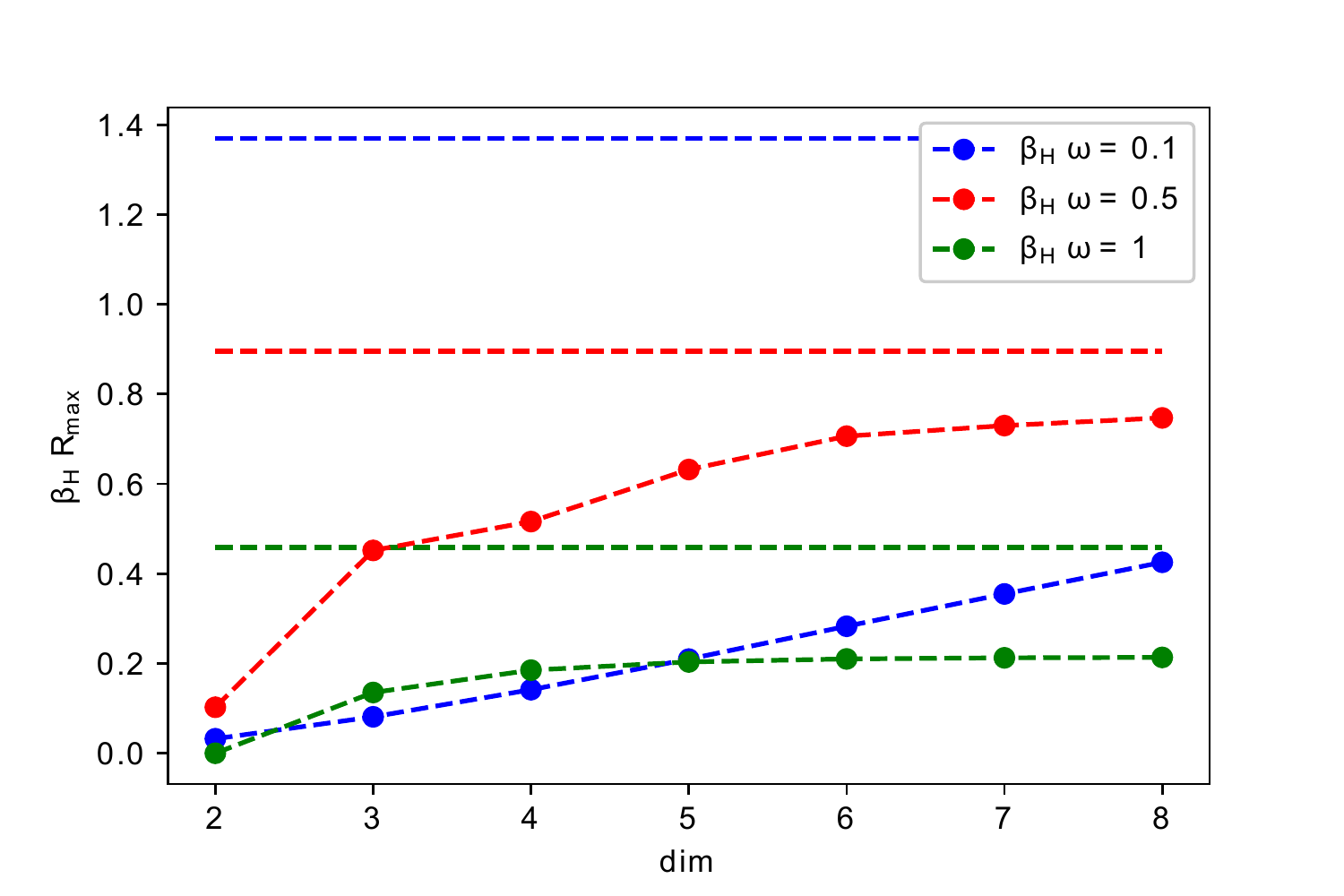}
    \includegraphics[width = 0.45 \textwidth]{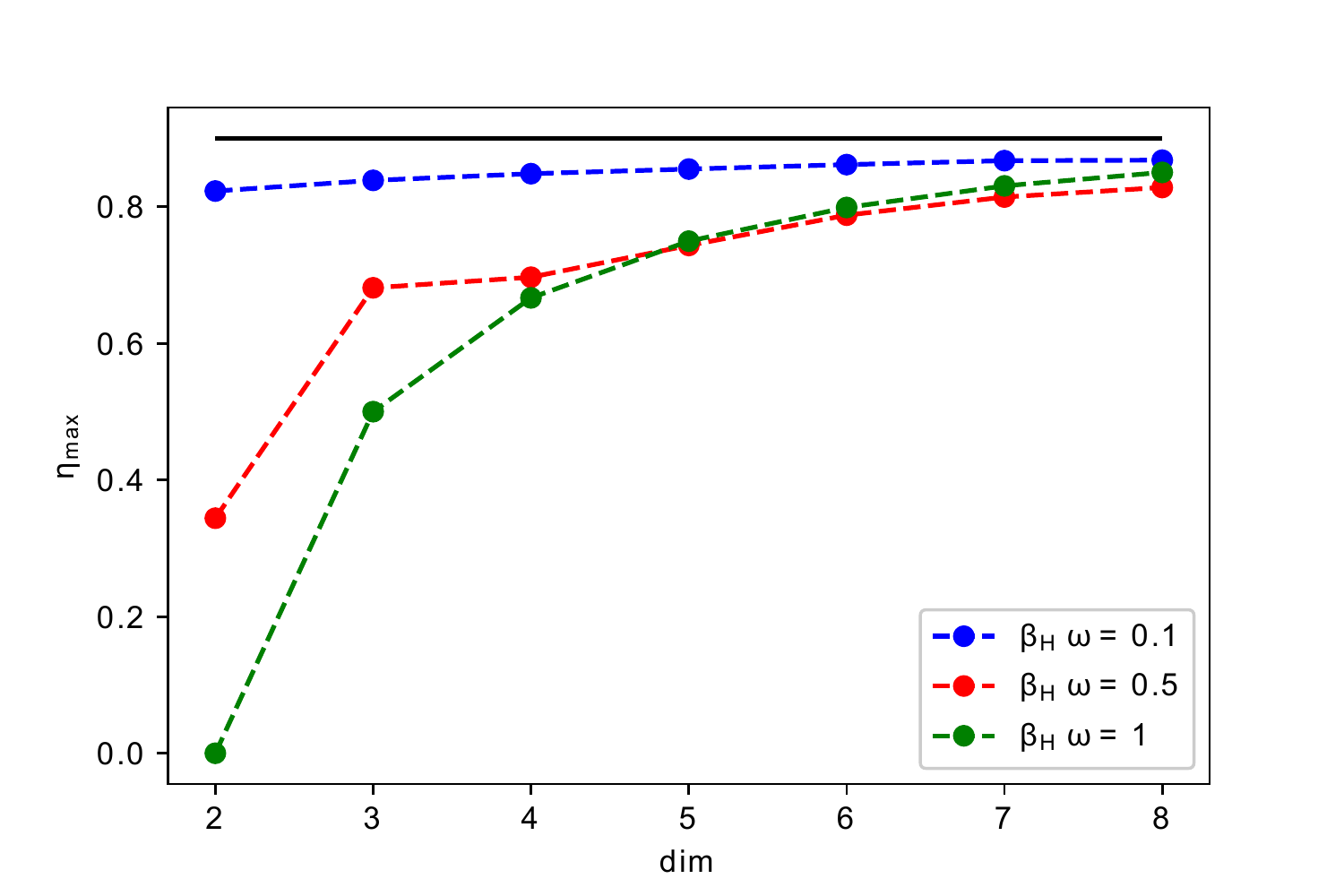}
    \caption{Work production and efficiency for the open-cycle engine with $d$-dimensional working body (with constant energy gap $\omega$). Plots are computed for a fixed initial Gibbs state in a cold (inverse) temperature $\beta_C \omega = 10$ and for different values of $\beta_H \omega$. \textit{Left panel:} Work production given by the amount of the optimal extracted ergotropy $\beta_H R_{max}$. The dashed horizontal lines correspond to the free energy bound \eqref{eq:BOEforGPTransformationthm}. \textit{Right panel:} The efficiency of the ergotropy extraction process \eqref{eq:Defn_efficiency}. The black solid line corresponds to the Carnot efficiency for $\beta_C/\beta_H = 10$. }
    \label{fig:qudit}
\end{figure}

\section{Outlook}\label{sec:5}


We have proposed a bound on ergotropy extraction, showed that it can be saturated for infinite dimensional systems, and,  for finite dimensional systems, we have investigated the gap between the bound and the value of extractable ergotropy. We use the idea of ergotropy extraction to design a stroke engine in the quantum regime, called open cycle engine. We prove that optimal work production and efficiency for these engines are achieved for extremal states of the thermal polytope of the initial state. We have  analytically characterized the optimal protocols for work production and efficiency for such engines in two and three dimensions and analyze higher dimensional cases numerically.

In the construction of open cycle engines we have assumed that the state of the working body is initialised in the Gibbs state, hence performance of the engine was not affected by coherences in the local Hamiltonian basis. However, understanding the role of coherence in ergotropy extraction in general remains an open problem. We may formalise it as follow: consider a working body in an initial state $\hat{\rho}$ (with coherences). Is it possible to achieve a state $\hat{\sigma}$ such that $R(\hat{\sigma})-R(\hat{\rho})>R(\mathcal{D}(\hat{\sigma}))-R(\mathcal{D}(\hat{\rho}))$ where $\mathcal{D}(\cdot)$ denotes dephasing\pmaz{?} If it was so, it would constitute a quantum advantage for ergotropy extraction, and may lead to constructions of heat engines with higher work yield per cycle than their classical counterpart.  

Secondly, it remains an open problem to investigate the relation between extractable ergotropy (optimized over all energy conserving operations) and the dimension of the working body. We have numerically attested this in Fig.~\ref{fig:saturation} for a qudit system initialized in the ground state being transformed via a thermal operation that exchanges maximum energy from bath.

On the general note, the bound (\ref{Alicki}) attests for the ergotropy of a composed system in the asymptotic limit of copies, singling out the Gibbs state as a completely passive state. Analogously, one could ask about extractable ergotropy for composed systems coupled with a heat bath, investigating the effect in which the emerging entanglement within the system influences the extractable ergotropy. Finally, to complete the characterization of the closed cycle engines \cite{Lobejko2020}, protocols for higher dimensional working body systems would need to be developed by optimizing ergotropy extraction, with cyclicity constraints taken into account.

Finally, a problem of finding physical systems saturating bounds derived with the help of thermal operations boils down in our case to the question if strong dependence of optimal values of efficiency and work production on system Hamiltonian and bath temperatures is present within a fixed physical implementation of an open-cycle heat engine. We leave as a subject of future research to check if this dependency results from identifying different energy preserving interactions as optimal in different temperature regimes.      

\subsection*{Acknowledgements}
TB acknowledges Ralph Silva, Matteo Lostaglio and A. de. Oliveira  for insightful discussions and comments during quantum thermodynamics summer school organized by ETH Zurich. We acknowledge support from the Foundation for Polish Science through IRAP project co-financed by EU within the Smart Growth Operational Programme (contract no.2018/MAB/5) and the National Science Centre, Poland, through grant SONATINA 2 2018/28/C/ST2/00364.

\appendix

\section{Supplementary information}
\subsection{Mathematical preliminaries of Thermal operation}\label{MPofTO}
In this section, we shall discuss the required mathematical properties of thermal operations and Gibbs preserving operations which have been used to analyse the ergotropy extraction protocols and the performance of open-cycle heat engines.  

\begin{defn}
\textbf{Thermal operation:} A complete positive trace-preserving (CPTP) map $\Phi$ is a thermal operation on the system in the state $\hat\rho$ if it transforms the product state of system and bath at equilibrium via an energy conserving unitary, followed by tracing out the bath: 
\begin{equation}\label{eq:TO_Again}
    \Phi(\hat\rho)=\Tr_{B}[{U\left(\hat\rho\otimes\hat\tau_{\beta}\right)U^{\dagger}}],
\end{equation}
where $U$ is a joint unitary commuting with the total Hamiltonian of the system and bath $[U, \hat{H}_S+ \hat{H}_B] = 0$, and $\hat\tau_{\beta}$ is a thermal state of the bath at some fixed inverse temperature $\beta$.
\end{defn}

Conservation of energy by the unitary can be understood as the first law of thermodynamics. The operation preserves the Gibbs state: $\Phi(\hat\gamma_{\beta})=\hat\gamma_{\beta}$.  
Next, we shall state the theorem which establishes the equivalence between thermal operation and Gibbs-preserving operation when they acts on a state that is diagonal in energy eigenbasis.
\begin{thm} 
[Equivalence between thermal operation and Gibbs-preserving operation \newline for a state that is diagonal in energy eigenbasis (\cite{Horodecki2013}] Consider a state $\hat\rho$ is diagonal in eigen basis of Hamiltonian $\hat H_S$ i.e., $\hat{\rho}=\sum_{i=1}^d\langle\epsilon_i|\hat \rho|\epsilon_i\rangle\dyad{\epsilon_i}$ where $\{|\epsilon\rangle\}_{i=1}^d$ are eigenvector of $\hat{H}_S$ with corresponding eigenvalues $\epsilon_i$. Then the set of achievable states from a given state $\hat\rho$ via thermal operation in presence of a bath at inverse temperature $\beta$, and Gibbs-preserving map that preserves $\hat\gamma_{\beta}$ are identical.
\end{thm}
Since in this work we want to analyze the extraction of ergotropy and optimal performance of open cycle engines, we focus on transformation of states which are diagonal in the energy eigenbasis. Setting up this framework, the central question to ask is what are the conditions on two states $\hat\rho$ and $\hat\sigma$ that are diagonal in the eigenbasis of Hamiltonian $\hat H_S$, such that they are related via a thermal operation $\Phi$:
\begin{equation}
    \Phi(\hat\rho)=\hat\sigma.
\end{equation}
By Lemma \ref{HO_IMP_LEM} from the main text of the paper, the transformation can be described by a stochastic matrix $\mathcal{A}_{\Phi}$ such that probability vectors $\v{p}$ and $\v{q}$ generated from the diagonal entries
\begin{eqnarray}
\v{p} = \Big(\langle \epsilon_1|\hat{\rho}|\epsilon_1\rangle  \ldots \langle \epsilon_d|\hat{\rho}|\epsilon_d\rangle \Big),\quad\quad
 \v{q} =  \Big(\langle \epsilon_1|\hat{\sigma}|\epsilon_1\rangle  \ldots \langle \epsilon_d|\tilde{\sigma}|\epsilon_d\rangle \Big),\label{p}
\end{eqnarray}
 satisfy
 $\mathcal{A}_{\Phi}(\v{p})=\v{q}$ and $\mathcal{A}_{\Phi}(\v{\gamma}_{\beta})=\v{\gamma}_{\beta}$, with 
\begin{equation}
     \v{\gamma}_{\beta} = \Big(\frac{e^{-\beta\epsilon_1}}{Z}  \ldots \frac{e^{-\beta\epsilon_d}}{Z}\Big).
\end{equation}
To describe the necessary and sufficient conditions for the existence of a stochastic matrix satisfying $\mathcal{A}_{\Phi}\v{p}=\v{q}$ and $\mathcal{A}_{\Phi}\v{\gamma}_{\beta}=\v{\gamma}_{\beta}$,  we proceed by introducing the concept of $\beta$-ordering of a population vector and thermo-majorization curves. These concepts are illustrated in Fig. \ref{Fig6}.

\begin{defn}
\textbf{Thermomajorization curves:} Define a vector $\v{s}=(1\;e^{-\beta\epsilon_{1}}\;e^{-\beta\epsilon_{2}}\;\ldots \;e^{-\beta\epsilon_{d-1}})$. Choose a permutation $\pi$ on \v{p} and \v{s} such that it leads to a non-increasing order of elements in the vector $\v{d}$ where $d_{k}=\frac{(\pi \v{p})_k}{ (\pi \v{s})_i }$, $k=0,\dots,d-1$. The set of points $\{\sum_{i=0}^{k}(\pi\v{p})_i,\sum_{i=0}^{k}(\pi\v{s})_k\}_{k=0}^{d-1}\cup\{0,0\}$, connected by straight lines, defines a curve associated with the state $\hat\rho$. We denote it by $\beta(\v{p})$ and call thermomajorization curve of the state $\hat\rho$ represented by $\v{p}$. 
\end{defn}

The points $\{\sum_{i=0}^{k}(\pi\v{p})_i,\sum_{i=0}^{k}(\pi\v{s})_i\}_{k=0}^{d-1}$ will be called elbows of the curve $\beta(\v{p})$. 
The curve is concave due to the non-increasing order of elements in $\v{d}$. Let us note that there might be more than one permutation leading to a concave curve $\beta(\v{p})$. 

\begin{defn}\textbf{$\beta$-order:} $\beta$-order of $\rho$ is defined as a vector $\pi(1\;\dots\;d)$.
\end{defn}\label{Defn:betaorder}
 $\beta$-order shows the order of segments assuring convexity of $\beta(\v{p})$. See Fig. \ref{termow} for examples of curves with different $\beta$-orders for a $d=3$ case. 

Thermomajorization curves are used to characterize possible transitions between states under Thermal Operations \cite{Horodecki2013} :

\begin{defn}\textbf{Thermomajorization:} A curve $\beta(\v{p})$ thermomajorizes $\beta(\v{q})$ iff all elbows of $\beta(\v{q})$ lie on $\beta(\v{p})$ or below it. 
\end{defn}

\begin{lemma}[Necessary and sufficient conditions for the existence of a thermal operation (\cite{Horodecki2013}]\label{betaorderingnecsuff}
A transition from a state $\hat\rho$ with diagonal entries $\v{p}$ to a state $\hat\sigma$ with diagonal entries $\v{q}$ under thermal operations is possible only if $\beta(\v{p})$ thermomajorizes $\beta(\v{q})$. If $\hat\rho$ is diagonal in the basis of local Hamiltonian, then the above condition is also sufficient.
\end{lemma}
Note that the set of all stochastic matrices which preserves the probability vector $\v{\gamma}_{\beta}$ forms a polytope (see \cite{MazurekPRA} for its complete characterization), thus the set of all achievable states also forms a polytope. We call this polytope a thermal polytope of $\hat\rho$, and denoted by $\mathcal{T}(\hat\rho)$. Its extremal points are defined with the help of the notion of \textit{tight thermomajorization}:

\begin{defn}\textbf{Tight thermomajorization:}
If the curve $\beta(\v{q})$ has all elbows on the curve $\beta(\v{p})$, then $\beta(\v{p})$ tightly thermomajorizes $\beta(\v{q})$.
\end{defn}

\begin{lemma}
[Characterizing the extremal points of $\mathcal{T}(\hat\rho)$ (Theorem 4 of \cite{MazurekPRA})]\label{ExtremalofTC} For a state $\hat\rho$ with probability vector $\v{p}$, a state $\hat\sigma$ with probability vector $\v{q}$ belongs to the set of extremal points of $\mathcal{T}(\hat\rho)$ if an only if $\beta(\v{p})$ tightly thermomajorizes $\beta(\v{q})$. 
\end{lemma}

\begin{figure}[ht]
\centering
\includegraphics[width = 0.5\linewidth]{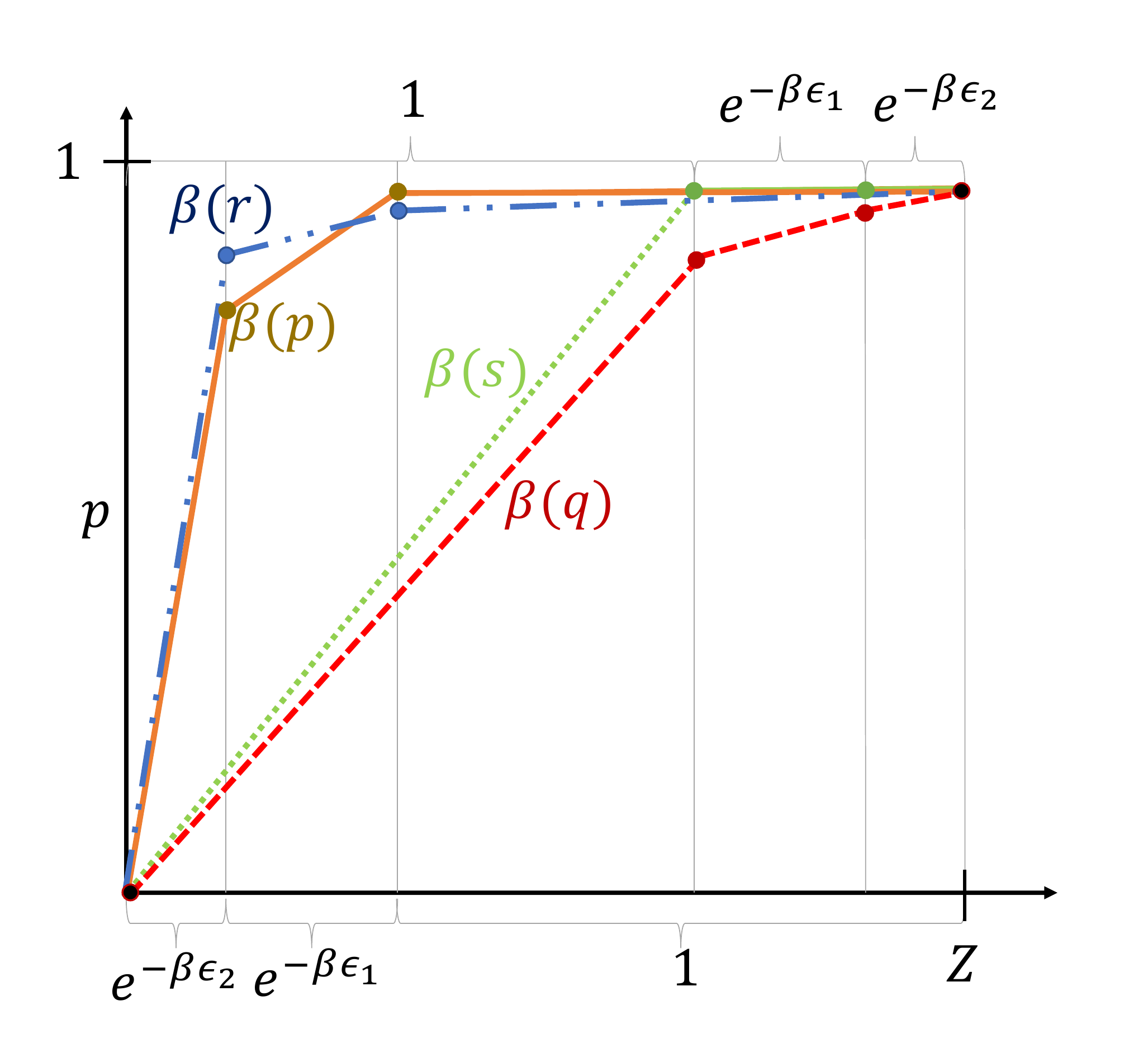}
\caption{\label{termow} 
 Thermomajorization diagram for a three dimensional system with Hamiltonian $\hat H=\epsilon_{1}|\epsilon_{1}\rangle\langle \epsilon_{1}|+\epsilon_{2}|\epsilon_{2}\rangle\langle \epsilon_{2}|$, $\epsilon_{2}>\epsilon_{1}$. Elbows of curves are indicated by circles. Curves $\beta(\v{s})$ and $\beta(\v{q})$ are \textit{thermomajorized} by curve $\beta(\v{p})$. Curve $\beta(\v{r})$ does not thermomajorize $\beta(\v{p})$, nor $\beta(\v{p})$  thermomajorizes $\beta(\v{r})$. Curve $\beta(\v{s})$ is \textit{tightly thermomajorized} by curve $\beta(\v{p})$. $\v{p}$ and $\v{r}$ have $\beta$-order $(3,2,1)$, while $\v{s}$, and $\v{q}$ have $\beta$-order $(1,2,3)$.\label{Fig6}}
\end{figure}

Finally, we will show that a thermal process which leads to descending $\beta$-order of the final state, i.e. $(d,d-1,\dots,1)$, maximizes the final energy, provided $\hat\sigma$ belongs to the set of extremal states of $\mathcal{T}(\hat\rho)$: 

\begin{lemma}\label{energy_max}
For every initial state $\hat\rho$ with probability vector $\v{p}$ in dimension $d$, among all states $\hat\sigma$, with probability vector $\v{q}$, which can be obtained from it under thermal operations, the state with the highest energy corresponds to thermomajorization curve $\beta(\v{q})$ which is tightly thermomajorized by $\beta(\v{p})$ and has $\beta$-order $(d,d-1,\dots,1)$.
\end{lemma}

\begin{proof}
As energy of the state depends solely on its diagonal, it can be deduced from the thermomajorization curve of the state (even in a case when some non-diagonal elements are initially present). A set of states achievable by thermal operations from a given initial state is a convex set. As energy is linear with respect to mixture of states, it is maximized for an extremal state in this set. From Lemma \ref{energy_max} and Lemma \ref{betaorderingnecsuff} we see that thermomajorozation curves of these extremal states are tightly thermomajorized by the curve of the initial state. Therefore, what remains to be shown is that the $\beta$-order of this curve should be $(d,d-1,\dots,1)$. Assume that the $\beta$-order of the final state is different, i.e. that in $(\alpha_{1},\dots,\alpha_{k},\alpha_{k+1},\dots,\alpha_{d})$ there exist $\alpha_{k}$ and $\alpha_{k+1}$ such that $\alpha_{k}<\alpha_{k+1}$. It is straightforward to see that swapping $\alpha_{k}$ with $\alpha_{k+1}$ and transitioning instead to a final state with $\beta$-order $(\alpha_{1},\dots,\alpha_{k+1},\alpha_{k},\dots,\alpha_{d})$ leads to a state with non-lower energy, as the modified $\beta$-order implies non-decreased occupation on level $\alpha_{k+1}$, and non-increased occupation on the level $\alpha_{k}$, compared to the previous final state, while level $\alpha_{k+1}$ is associated with higher energy than $\alpha_{k}$, according to the assumption $\alpha_{k+1}>\alpha_{k}$. We can use the above argument to transit in steps from an arbitrary $\beta$-order to $(d,d-1,\dots,1)$, while having the energy of the final state unchanged or increasing in each step.   
\end{proof}

We use the above lemma to construct states of maximal energy in the thermal cone of the ground state of finite dimensional system, and analyze their ergotropy in Fig. \ref{fig:saturation}.

\subsection{Proof of Lemma \ref{Postive_Heat}}\label{PROOF_positive_heat_flow}
We begin by re-writing the definition of the amount of heat that is exchanged with  hot reservoir given in Eq. \eqref{Defn:HeatExchanged} 
\begin{equation}\label{Defn:HeatExchangedAgain}
    Q(\mathcal{E}_{H}(\hat{\rho}_S))= \Tr(\hat{H}_S\big(\mathcal{E}_{H}(\hat{\rho}_S)-\hat{\rho}_S\big)),
\end{equation}
where $\hat{\rho}_S$ is the state of the working body  that has been thermalized with cold reservoir at inverse temperature $\beta_C$ as given in Eq. \eqref{state_of_WB}. We see that function $Q(\cdot)$ is linear.  Therefore, when we minimise the function $Q(\cdot)$ over all the states that can be achieved from $\hat\rho_S$ via thermal operation i.e.,
\begin{equation}
    \min_{\hat\sigma\in\mathcal{T}(\hat\rho_S)}Q(\hat\sigma_S) := Q(\hat\sigma_S^{*}),
\end{equation}
where
the minimum is achieved at a extremal point labelled by $\hat\sigma_S^{*}$ of the set $\mathcal{T}(\hat\rho_S)$. Therefore, it is enough to prove $Q(\hat\sigma_S^{*})\geq 0$ which we shall do by showing $\hat\sigma_S^{*}=\hat\rho_S$. From lemma \ref{HO_IMP_LEM}, we see that applying thermal operation on the state $\hat\rho_S$ in presence of a hot bath at inverse temperature $\beta_H$, is equivalent to applying a stochastic matrix on the probability vector   
\begin{equation}
    \v{p}_S = \Big(\langle \epsilon_1|\hat{\rho}_S|\epsilon_1\rangle  \ldots \langle \epsilon_d|\hat{\rho}_S|\epsilon_d\rangle \Big),
\end{equation}
such that the stochastic matrix preserves
\begin{equation}\label{eq:HotGibbs}
     \v{\gamma}_{\beta_H} = \Big(\frac{e^{-\beta_H\epsilon_1}}{Z}  \ldots \frac{e^{-\beta_H\epsilon_d}}{Z}\Big).
\end{equation}
Note that the probability vector $\v{p}_S$ has $\beta$-order $(12\ldots d)$. From lemma \ref{betaorderingnecsuff}, we see that $\v{p}_S\succ_{\beta_H}\v{s}^{*}_S$ where
\begin{equation}
    \v{s}^{*}_S = \Big(\langle \epsilon_1|\hat\sigma_S^{*}|\epsilon_1\rangle  \ldots \langle \epsilon_d|\hat\sigma_S^{*}|\epsilon_d\rangle \Big).
\end{equation}
If $\v{s}^{*}_S$ has $\beta$-order $(12\ldots d)$ then using lemma \ref{ExtremalofTC}, immediately gives $\v{s}^{*}_S=\v{p}_S$.\\
Now consider the situation when $\v{s}^{*}_S$ has $\beta$-order $(\alpha_1\ldots\alpha_d)\neq(12\ldots d)$. Then, there will be at least two conseceutive entries $\alpha_{k}$ and $\alpha_l$ in the $\beta$-order $(\alpha_1\ldots\alpha_d)$ such that $\alpha_{k}>\alpha_l$. Let's call $\alpha_{k}=i$ and $\alpha_l=j$ where $i\;,\; j\in\{1,\ldots,d\}$. Since $i>j$, therefore $\epsilon_i>\epsilon_j$ since conventionally we take $\epsilon_1\leq\epsilon_2\leq\ldots\epsilon_d$. Now we shall construct the following stochastic matrix 
\begin{equation}
    \mathcal{A}=\begin{pmatrix}
    1-e^{-\beta_H(\epsilon_i-\epsilon_j)}& 1\\
    e^{-\beta_H(\epsilon_i-\epsilon_j)} & 0
    \end{pmatrix} \oplus \mathbb{I}_{\text{Rest}},
\end{equation}
where $\mathbb{I}_{\text{Rest}}$ is the identity matrix on the subspace orthogonal to the subspace spanned by $\{|\epsilon_i\rangle,|\epsilon_j\rangle\}$. It is easy to see that $\mathcal{A}(\v{\gamma}_{\beta_H})=\v{\gamma}_{\beta_H}$ where $\v{\gamma}_{\beta_H}$ is defined in Eq. \eqref{eq:HotGibbs}.
Next, from the definition of $\beta$-order given in \ref{Defn:betaorder}, we can write 
\begin{equation}\label{Eq:betaineq}
    \frac{\langle \epsilon_i|\hat{\sigma}^{*}_S|\epsilon_i\rangle}{e^{-\beta_H\epsilon_i}} > \frac{\langle \epsilon_j|\hat{\sigma}^{*}_S|\epsilon_j\rangle}{e^{-\beta_H\epsilon_j}},
\end{equation}
since $i=\alpha_k$ appears before $j=\alpha_l$ in the $\beta$-order $(\alpha_1\ldots\alpha_d)$. Note that we can always choose $\alpha_k$ and $\alpha_l$ in such a way that Eq. \eqref{Eq:betaineq} is a strict inequality. If it is not true, this will lead us to conclude that $\v{s}_S^*=\v{\gamma}_{\beta_H}$ which is not the case, because $\sigma_S^*$ has to be the extremal point of $\mathcal{T}(\hat\rho_S)$. From Eq. \eqref{Eq:betaineq} we have,
\begin{equation}\label{eq:Population_reduction}
    \langle \epsilon_i|\hat{\sigma}^{*}_S|\epsilon_i\rangle > \langle \epsilon_j|\hat{\sigma}^{*}_S|\epsilon_j\rangle e^{-\beta_H(\epsilon_i-\epsilon_j)},
\end{equation}
Now, applying the stochastic matrix $\mathcal{A}$ on $\v{s}_S^{*}$ transforms the probability on the energy level $|\epsilon_j\rangle$ and $|\epsilon_i\rangle$ as 
\begin{eqnarray}
(\mathcal{A}\v{s}_S^{*})_{j} &=& (\langle \epsilon_j|\hat{\sigma}^{*}_S|\epsilon_j\rangle + \langle \epsilon_i|\hat{\sigma}^{*}_S|\epsilon_i\rangle)-\langle \epsilon_j|\hat{\sigma}^{*}_S|\epsilon_j\rangle e^{-\beta_H(\epsilon_i-\epsilon_j)},\\
(\mathcal{A}\v{s}_S^{*})_{i} &=& \langle \epsilon_j|\hat{\sigma}^{*}_S|\epsilon_j\rangle e^{-\beta_H(\epsilon_i-\epsilon_j)},
\end{eqnarray}
and leaves the probability on the other level unchanged. As $\epsilon_i > \epsilon_j$ , employing \eqref{eq:Population_reduction}, we immediately have a state with diagonal entries given by $\mathcal{A}\v{s}_S^{*}$ has strictly less average energy than $\hat\sigma_S^{*}$, that can be via a thermal operation $\hat\sigma_S^{*}$. To put it simply , using thermal process $\mathcal{A}$ that acts non-trivially only on the energy levels $|\epsilon_i\rangle$ and $|\epsilon_j\rangle$, we have decreased the population on the higher energy level  $|\epsilon_i\rangle$ , and increased the population in the lower energy level $|\epsilon_j\rangle$. Therefore, the state with diagonal entries given by $\mathcal{A}\v{s}_S^{*}$, has strictly less energy than $\hat\sigma_S^{*}$ which is contradicting with the assumption that $\hat\sigma^*$ is the sate with lower energy in the set $\mathcal{T}(\hat\rho_S)$. Therefore, we conclude that $\hat\sigma_S^{*}$ has to be equals to $\hat{\rho}_S$. 

\subsection{Explicit work storage battery modelled by translational invariant weight}\label{SubWorkStotage}

In this section, we introduce a battery model that can be use to store the ergotropy present in a system as work. This battery is modelled by an ideal weight \cite{Aberg2014, Skrzypczyk2014} with translational invariant symmetry.  It is  described by the Hamiltonian
\begin{equation} \label{weight_hamiltonian}
    \hat{H}_B = \sum_{n=-\infty}^{\infty}n\omega |n\rangle\langle n|,
\end{equation}
while the unitaries $\hat{U}_{B}$ are restricted to these satisfying translational invariant symmetry given by 
\begin{equation} \label{translation_invariance}
    [\hat U_{B}, \hat \Gamma_\epsilon] = 0,
\end{equation}
where $\hat \Gamma_\epsilon$ is a shift operator which displaces the energy spectrum of the weight, i.e. $\hat \Gamma_\epsilon^\dag \hat H_B \hat \Gamma_\epsilon = \hat H_B + \epsilon$, and $\epsilon$ is an arbitrary real constant. The above assumption implies that the resulting transformation on the battery does not depend on its initial state, a welcomed property for its iterative use. 

It was shown that work defined as the change of average energy of the battery, resulting from the unitaries with the above constraint, satisfies the Second Law of Thermodynamics \cite{Skrzypczyk2014}. Moreover, for the case of the battery state having no coherences in the eigenbasis of $\hat{H}_{B}$, it was shown in  in \cite{Aberg2013, Alhambra2016, Lobejko2021} that the average energy gain corresponds to ergotropy of the working body.

\subsection{Detailed analysis of open cycle heat engine with qutrit working body.}\label{SubDet}
In this section we provide the detailed calculation of the amount of extractable ergotropy from the hot bath for the initial state 
that is given by 
\begin{equation} \hat{\rho}_S=\dyad{0}{0}+q^C_{10}\dyad{1}{1}+q^C_{20}\dyad{2}{2}),
\end{equation}
where $Z_C=1+q^C_{10}+q^C_{20}$ where and $q^C_{ij}$ is introduced in Eq. \eqref{eq:Notation} as $q^C_{ij} = e^{-\beta_C(\omega_i-\omega_j)}$. We write the probability vector generated from the diagonal of $\hat{\rho}_S$ as
\begin{equation}\label{eq:Initialstate2}
    \v{p}_S = \frac{1}{Z_C}(1 \quad q^C_{10} \quad q^C_{20}),
\end{equation}
We can immediately see the $\beta$-order of the of the probability vector at inverse temperature $\beta_H$ is $\v{p}_S$ is (123) since
\begin{equation}
    \frac{1}{Z_C}\geq \frac{q^C_{10}}{e^{-\beta_H\omega_1}} \geq \frac{q^C_{20}}{e^{-\beta_H\omega_2}},
\end{equation}

Now we shall employ the lemma from \cite{Mazurek_2018}, that characterizes the extremal thermal processes that act on any probability vector $\v{p}$ with $\beta$-order $(123)$ and give the spectrum of the extremal states of $\mathcal{T}(\hat{\rho}_S)$.  
\begin{lemma}\label{extremal}
[Based on Table 1 of \cite{Mazurek_2018}] For any probability vector $\v{p}$ with $\beta$-order $(123)$, thermal processes at the inverse temperature $\beta_H$ that gives the spectrum of non-trivial extremal states in $\mathcal{T}(\hat{\rho}_S)$ are given by 
\begin{eqnarray}\label{eq:ProcessesThm}
\begin{cases} \label{eq:ExtremalTP}
       A_1, A_2, A_5, A_9 &\quad \mathrm{for~} \beta_H\geq\beta_0 ,\\
       A_1, A_2, A_5, A_{12}, A_{13} &\quad\mathrm{for~}\beta_H<\beta_0,\\
     \end{cases}
\end{eqnarray}  
 where $\beta_0$ is defined by the following relation 
 \begin{eqnarray}
 e^{-\beta_0\omega_1}+e^{-\beta_0\omega_2}=1,
 \end{eqnarray}
 and 
\begin{eqnarray}
\label{A6}
A_{1} &=& \begin{pmatrix}
1-q^H_{10}& 1 & 0 \\
q^H_{10} & 0 & 0 \\
0 & 0 & 1 
\end{pmatrix}\quad \quad \quad \quad \quad\quad\quad\quad
A_{2} = \begin{pmatrix}
1 & 0 & 0 \\
0 & 1-q^H_{21} & 1 \\
0 & q^H_{21} & 0 
\end{pmatrix},\\
A_{5} &=& \begin{pmatrix}
1-q^H_{20}& q^H_{21} & 0 \\
0 & 1-q^H_{21} & 1 \\
q^H_{20} & 0 & 0 
\end{pmatrix}\quad \quad \quad\quad\quad\;\;
A_{9} =\begin{pmatrix}
1-q^H_{10}-q^H_{20}& 1 & 1 \\
q^H_{10} & 0 & 0 \\
q^H_{20} & 0 & 0 
\end{pmatrix},\label{A9}\\
A_{12} &=& 
\begin{pmatrix}
0 & q^H_{01}-q^H_{21} & 1 \\
q^H_{10} & 0 & 0 \\
1-q^H_{10} & 1-q^H_{01}+q^H_{21} & 0 
\end{pmatrix}\quad\quad\quad
A_{13} = 
\begin{pmatrix}
0 & q^H_{01}-q^H_{21} & 1 \\
1-q^H_{20} & 1-q^H_{01}+q^H_{21} & 0 \\
q^H_{20} & 0 & 0 
\end{pmatrix}.\label{A13}
\end{eqnarray}
 Moreover, for the probability vector with $\beta$- order $(123)$, the final $\beta$-order for the thermal process given in Eq. \eqref{eq:ExtremalTP} is given by 
 \begin{center}\label{Tableofbetaorder}
\begin{tabular}{ |c|c|c|c| } 
\hline
Thermal processes & $\beta$-order \\
\hline
 $A_1$ & (213)  \\ 
$A_2$& (132)  \\ 
$A_5$& (312)  \\ 
$A_9$& (231) \text{  or  } (321)  \\ 
\hline
$A_{12}$& (231) \\
$A_{13}$& (321) \\
\hline
\end{tabular}
\end{center}
\end{lemma}

Using this lemma, we shall analyze the optimal performance of open cycle heat engine with initial state having spectrum $\v{p}_S$ i.e., calculate the optimal ergotropy that can be extracted from the hot bath via thermal process $A_i$ and stored it as work. We denote it as $R_{A_i}(\v{p}_S)$. Thus we can calculate the efficiency of such heat engine.\\

\textbf{Calculation of the amount of transferred heat from Hot bath  and extractable ergotropy from $\v{p}_S$ using thermal process $A_1$: }

Note that thermal process $A_1$ transforms the initial probability vector given in Eq. \eqref{eq:Initialstate2} in the following way
\begin{eqnarray}\label{eq:TransitionA1}
&&\frac{1}{Z_C}\begin{pmatrix}
1-q^H_{10}& 1 & 0 \\
q^H_{10} & 0 & 0 \\
0 & 0 & 1 
\end{pmatrix}
\begin{pmatrix}
1 \\
q^C_{10} \\
q^C_{20} 
\end{pmatrix}=
\frac{1}{Z_C}\begin{pmatrix}
(1-q^{H}_{10})+q^C_{10} \\
q^H_{10} \\
q^C_{20} 
\end{pmatrix} = \text{Spec}(\hat{\rho}^1_{S}),
\end{eqnarray}
From the table \ref{Tableofbetaorder}, we have the $\beta$- order of the final probability vector in Eq. \eqref{eq:TransitionA1}  is (213). That means 
\begin{eqnarray}\label{eq:betaorder_A1}
&&\frac{q^H_{10}}{Z_Cq^H_{10}}\geq\frac{(1-q^{H}_{10})+q^C_{10}}{Z_C} \geq \frac{q^C_{20}}{Z_Cq^H_{20}}\nonumber\\
&\Rightarrow& (1-q^{H}_{10})+q^C_{10} \geq \frac{q^C_{20}}{q^H_{20}} \geq q^C_{20},
\end{eqnarray}
where the last inequality follows from the fact $0\leq q^H_{20}\leq1$. Therefore, we have the following order of population that is compatible with Eq. \eqref{eq:betaorder_A1} is given by 
\begin{eqnarray}
&&\frac{(1-q^{H}_{10})+q^C_{10}}{Z_C} \geq \frac{q^H_{10}}{Z_C} \geq \frac{q^C_{20}}{Z_C}\label{eq:TrivialOrderA1},\\
&& \frac{q^H_{10}}{Z_C} \geq \frac{(1-q^{H}_{10})+q^C_{10}}{Z_C} \geq \frac{q^C_{20}}{Z_C}\label{eq:NonTrivialOrderA1}.
\end{eqnarray}
Note that if the order of the population given by Eq. \eqref{eq:TrivialOrderA1}, then no ergotropy extraction is possible from the final state because it is passive. Finally, if the order of the population of final probability vector in Eq. \eqref{eq:TransitionA1} is given by Eq. \eqref{eq:NonTrivialOrderA1} then the amount ergotropy that can be extracted is given by 
\begin{eqnarray}
R_{A_1}(\v{p}_S)&=&\max\Big\{0,\frac{1}{Z_C}\Big(q^H_{10}\omega_1+q^C_{20}\omega_2-q^C_{20}\omega_2-((1-q^{H}_{10})+q^C_{10})\omega_1\Big)\Big\}\nonumber\\
&=&\max\Big\{0,\frac{\omega_1}{Z_C}\Big(2q^{H}_{10}-1-q^C_{10}\Big)\Big\}.
\end{eqnarray}
Thus, a condition for positive ergotropy extraction is given by 
\begin{equation}
    2q^H_{10}\geq 1+ q^{C}_{10}.
\end{equation}
\textbf{Calculation of the amount of transferred heat from Hot bath  and extractable ergotropy from $\v{p}_S$ using thermal process $A_2$:} Thermal process $A_2$ transforms the initial probability vector given in Eq. \eqref{eq:Initialstate2} in the following way
\begin{eqnarray}\label{eq:TransitionA2}
\frac{1}{Z_C}\begin{pmatrix}
1& 0 & 0 \\
0 & 1-q^H_{21} & 1 \\
0 & q^H_{21} & 0 
\end{pmatrix}
\begin{pmatrix}
1 \\
q^C_{10} \\
q^C_{20} 
\end{pmatrix}=
\frac{1}{Z_C}\begin{pmatrix}
1 \\
(1-q^H_{21})q^C_{10}+q^C_{20} \\
q^H_{21}q^C_{10} 
\end{pmatrix}=\text{Spec}(\hat{\rho}^2_{S}).
\end{eqnarray}
From the table \ref{Tableofbetaorder}, we have the $\beta$- order of the final probability vector in Eq. \eqref{eq:TransitionA2}  is (132). That means 
\begin{eqnarray}\label{eq:betaorder_A2}
&&\frac{1}{Z_C}\geq\frac{q^H_{21}q^C_{10}}{Z_Cq^H_{20}}\geq\frac{(1-q^{H}_{21})q^C_{10}+q^C_{20}}{Z_Cq^{H}_{10}}, \\
&\Rightarrow& 1 \geq q^H_{21}q^C_{10} \quad  \text{and} \quad 1 \geq (1-q^{H}_{21})q^C_{10}+q^C_{20},\nonumber
\end{eqnarray}
where the last inequalities follows from $0\leq q^H_{10} \leq q^H_{20}\leq1$. Thus, order of population that is compatible with Eq. \eqref{eq:betaorder_A2} is given by
\begin{eqnarray}
&&\frac{1}{Z_C}\geq \frac{(1-q^{H}_{21})q^C_{10}+q^C_{20}}{Z_C} \geq \frac{q^{H}_{21}q^{C}_{10}}{Z_C},\label{eq:TrivialOrderA2}\\
&&\frac{1}{Z_C}\geq \frac{q^{H}_{21}q^{C}_{10}}{Z_C} \geq \frac{(1-q^{H}_{21})q^C_{10}+q^C_{20}}{Z_C},\label{eq:NonTrivialOrderA2}
\end{eqnarray}
If the order of the population of the final probability vector in Eq. \eqref{eq:TransitionA2} is given by Eq. \eqref{eq:TrivialOrderA2}, then non-zero ergotropy extraction is not possible due to passivity of final probability vector. Thus, amount of ergotropy that can be extracted from the final probability vector in Eq. \eqref{eq:TransitionA2} using thermal operation $A_2$ is given by 
\begin{eqnarray}
R_{A_2}(\v{p}_S)&=& \max\Big\{0,\frac{1}{Z_C}(\omega_2-\omega_1)\Big(2q^{H}_{21}q^{C}_{10}-(q^{C}_{10}+q^{C}_{20}\Big)\Big\}.\nonumber
\end{eqnarray}
Therefore, a condition for positive ergotropy extraction can be expressed as
\begin{equation}
    2q^{H}_{21}\geq 1+q^{C}_{21}.
\end{equation}
\textbf{Calculation of the amount of transferred heat from Hot bath  and extractable ergotropy from $\v{p}_S$ using thermal process $A_5$:}
Thermal process $A_5$ transforms the initial probability vector given in Eq. \eqref{eq:Initialstate2} as follows 
\begin{eqnarray}
\label{eq:TransitionA5}
&&\frac{1}{Z_C}\begin{pmatrix}
1-q^{H}_{20}& q^{H}_{21} & 0 \\
0 & 1-q^{H}_{21} & 1 \\
q^{H}_{20} & 0 & 0 
\end{pmatrix}
\begin{pmatrix}
1 \\
q^C_{10} \\
q^C_{20} 
\end{pmatrix}=
\frac{1}{Z_C}\begin{pmatrix}
1-q^{H}_{20}+q^{H}_{21}q^{C}_{10} \\
(1-q^H_{21})q^C_{10}+q^C_{20} \\
q^H_{20} 
\end{pmatrix}=\text{Spec}(\hat{\rho}^3_{S}).
\end{eqnarray}
From the table \ref{Tableofbetaorder}, we have the $\beta$- order of the final probability vector in Eq. \eqref{eq:TransitionA5}  is (312). That means 
\begin{eqnarray}\label{eq:betaorder_A5}
&&\frac{q^H_{20}}{Z_C q^H_{20}}\geq \frac{1-q^{H}_{20}+q^{H}_{21}q^{C}_{10}}{Z_C}\geq \frac{(1-q^H_{21})q^C_{10}+q^C_{20}}{Z_C q^H_{10}},\nonumber\\
&\Rightarrow& (1-q^H_{21})q^C_{10}+q^C_{20} \geq (1-q^H_{21})q^C_{10}+q^C_{20},
\end{eqnarray}
where we write the last inequality using the fact $0\leq q^H_{10}\leq 1$. Therefore, the order of populations that are compatible with Eq. \eqref{eq:betaorder_A5} is given by 
\begin{eqnarray}
&&\frac{1-q^{H}_{20}+q^{H}_{21}q^{C}_{10}}{Z_C}\geq \frac{(1-q^H_{21})q^C_{10}+q^C_{20}}{Z_C}\geq \frac{q^H_{20}}{Z_C},\label{eq:TrivialOrderA5}\\
&&\frac{1-q^{H}_{20}+q^{H}_{21}q^{C}_{10}}{Z_C} \geq \frac{q^H_{20}}{Z_C} \geq \frac{(1-q^H_{21})q^C_{10}+q^C_{20}}{Z_C},\label{eq:1NonTrivialOrderA5}\\
&&\frac{q^H_{20}}{Z_C} \geq \frac{1-q^{H}_{20}+q^{H}_{21}q^{C}_{10}}{Z_C} \geq \frac{(1-q^H_{21})q^C_{10}+q^C_{20}}{Z_C},\label{eq:2NonTrivialOrderA5}
\end{eqnarray}
If the order of the population for the final probability vector in Eq. \eqref{eq:TransitionA5} is given by Eq. \eqref{eq:TrivialOrderA5}, then no ergotropy extraction is possible since final probability vector is already passive. A non-trivial ergotropy extraction is possible if the ordering is given by Eq.~\eqref{eq:1NonTrivialOrderA5} and Eq. \eqref{eq:2NonTrivialOrderA5}. Therefore, we calculate the amount of extractable ergotropy as 
\begin{eqnarray}\label{eq:ErgA5}
R_{A_5}(\v{p}_S)&=&\max\Big\{0, \frac{(\omega_{2}-\omega_{1})}{Z_C}\Big(q^{H}_{20}-q^{C}_{20}-(1-q^{H}_{21})q^{C}_{10}\Big), \;\frac{\omega_{2}}{Z_C}q^{H}_{20}-\frac{\omega_{1}}{Z_C}\Big((1-q^{H}_{20})\nonumber\\&+&q^{H}_{21}q^{C}_{10}\Big)-\frac{(\omega_{2}-\omega_{1})}{Z_C}\Big((1-q^{H}_{21})q^{C}_{10}+q^{C}_{20}\Big)\Big\}.\nonumber\\
\end{eqnarray}

\textbf{Calculation of the amount of transferred heat from Hot bath  and extractable ergotropy from $\v{p}_S$ using thermal process $A_9$:}
Thermal process $A_9$ transforms the initial probability vector given in Eq. \eqref{eq:Initialstate2} as follows 
\begin{eqnarray}
\label{eq:TransitionA9}
&&\frac{1}{Z_C}\begin{pmatrix}
1-q^{H}_{10}-q^{H}_{20}& 1 & 1 \\
q^{H}_{10} & 0 & 0 \\
q^{H}_{20} & 0 & 0 
\end{pmatrix}
\begin{pmatrix}
1 \\
q^C_{10} \\
q^C_{20} 
\end{pmatrix}=
\frac{1}{Z_C}\begin{pmatrix}
1+q^{C}_{10}+q^{C}_{20}-q^{H}_{10}-q^{H}_{20} \\
q^H_{10} \\
q^H_{20} 
\end{pmatrix}=\text{Spec}(\hat{\rho}^4_{S}).
\end{eqnarray}
Since $q^H_{10}\geq q^H_{20}$ is true for any plausible values of $\beta_H$, therefore the possible orders of population is given by
\begin{eqnarray}
&&\frac{1+q^{C}_{10}+q^{C}_{20}-q^{H}_{10}-q^{H}_{20}}{Z_C}\geq \frac{q^H_{10}}{Z_C} \geq \frac{q^H_{20}}{Z_C}\label{eq:TrivialOrderA9},\\
&&\frac{q^H_{10}}{Z_C}\geq \frac{1+q^{C}_{10}+q^{C}_{20}-q^{H}_{10}-q^{H}_{20}}{Z_C}\geq \frac{q^H_{20}}{Z_C}\label{eq:1NonTrivialOrderA9},\\
&&\frac{q^H_{10}}{Z_C}\geq \frac{q^H_{20}}{Z_C}\geq \frac{1+q^{C}_{10}+q^{C}_{20}-q^{H}_{10}-q^{H}_{20}}{Z_C}.\label{eq:2NonTrivialOrderA9}
\end{eqnarray}
If the ordering of the populations in the final probability vectors are given by Eq. \eqref{eq:TrivialOrderA9}, then no ergotropy extraction is possible since initial state is already passive. Thus non-trivial ergotropy extraction is possible if the ordering of the population in the final probability vector in Eq. \eqref{eq:TransitionA9}  is given by Eq. \eqref{eq:1NonTrivialOrderA9} and Eq. \eqref{eq:2NonTrivialOrderA9}. Thus, the amount of extractable ergotropy is given by
\begin{eqnarray}\label{eq:ErgA9}
R_{A_9}(\v{p}_S) &=& \max\Big\{0,\;\frac{\omega_{1}}{Z_C}\Big(2q^{H}_{10}+q^{H}_{20}-(1+q^{C}_{10}+q^{C}_{20})\Big),\nonumber\\ &&\frac{\omega_1}{Z_C}(q^{H}_{10}-q^{H}_{20})+\frac{\omega_{2}}{Z_C}\Big(2q^{H}_{20}+q^{H}_{10}-(1+q^{C}_{10}+q^{C}_{20})\Big)\Big\}.
\end{eqnarray}

\textbf{Calculation of the amount of transferred heat from Hot bath  and extractable ergotropy from $\v{p}_S$ using thermal process $A_{12}$:} Thermal process $A_{12}$ transforms the initial probability vector given in Eq. \eqref{eq:Initialstate2} as follows 
\begin{eqnarray}
\label{eq:TransitionA12}
&&\frac{1}{Z_C}\begin{pmatrix}
0 & q^H_{01}-q^H_{21} & 1 \\
q^H_{10} & 0 & 0 \\
1-q^H_{10} & 1-q^H_{01}+q^H_{21} & 0 
\end{pmatrix}
\begin{pmatrix}
1 \\
q^C_{10} \\
q^C_{20} 
\end{pmatrix}=
\frac{1}{Z_C}\begin{pmatrix}
(q^{H}_{01}-q^{H}_{21})q^{C}_{10}+q^{C}_{20} \\
q^H_{10} \\
(1-q^{H}_{10})+(1-q^{H}_{01}+q^{H}_{21})q^C_{10} 
\end{pmatrix}=\text{Spec}(\hat{\rho}^5_{S}),\nonumber\\
\end{eqnarray}
From the table \ref{Tableofbetaorder}, we have the $\beta$- order of the final probability vector in Eq. \eqref{eq:TransitionA12}  is (231). That implies
\begin{eqnarray}\label{eq:betaorder_A12}
\frac{q^H_{10}}{Z_Cq^H_{10}}&\geq& \frac{(1-q^{H}_{10})+(1-q^{H}_{01}+q^{H}_{21})q^C_{10}}{Z_Cq^H_{20}}\nonumber\\&\geq& \frac{(q^{H}_{01}-q^{H}_{21})q^{C}_{10}+q^{C}_{20}}{Z_C}\nonumber\\
&\Rightarrow& q^{H}_{10}\geq \frac{q^{H}_{10}}{q^{H}_{20}}(1-q^{H}_{10})+(1-q^{H}_{01}+q^{H}_{21})q^C_{10} \nonumber\\&\geq& (1-q^{H}_{10})+(1-q^{H}_{01}+q^{H}_{21})q^C_{10},
\end{eqnarray}
where the last inequality follows from the fact $q^{H}_{10}\geq q^{H}_{20}$. Hence the ordering of the populations of the final probability vector Eq. \eqref{eq:TransitionA12} that are compatible with Eq. \eqref{eq:betaorder_A12}  is given by 
\begin{eqnarray}
&&\frac{1}{Z_C}\Big((q^{H}_{01}-q^{H}_{21})q^{C}_{10}+q^{C}_{20}\Big)\geq \frac{q^{H}_{10}}{Z_C}\geq \frac{1}{Z_C}\Big((1-q^{H}_{10})+(1-q^{H}_{01}+q^{H}_{21})q^C_{10}\Big),\label{eq:TrivialorderA12}\\
&&\frac{q^{H}_{10}}{Z_C}\geq \frac{1}{Z_C}\Big((q^{H}_{01}-q^{H}_{21})q^{C}_{10}+q^{C}_{20}\Big)\geq\frac{1}{Z_C}\Big((1-q^{H}_{10})+(1-q^{H}_{01}+q^{H}_{21})q^C_{10},\label{eq:1NonTrivialorderA12}\\
&& \frac{q^{H}_{10}}{Z_C}\geq \frac{1}{Z_C}\Big((1-q^{H}_{10})+(1-q^{H}_{01}+q^{H}_{21})q^C_{10}\Big)\geq\frac{1}{Z_C}\Big((q^{H}_{01}-q^{H}_{21})q^{C}_{10}+q^{C}_{20}\Big),\label{eq:2NonTrivialorderA12}
\end{eqnarray}
If the ordering of population of the final probability vector in Eq. \eqref{eq:TransitionA12} is given by Eq. \eqref{eq:TrivialorderA12}, then it is not possible to extract ergotropy from final probability vector due to passivity. Therefore, non-trivial ergotropy extraction is possible only if the ordering of populations in final probability vector is given by Eq. \eqref{eq:1NonTrivialorderA12} or Eq. \eqref{eq:2NonTrivialorderA12}. Thus, the amount of ergotropy that can be extracted using thermal operation $A_{12}$ is
\begin{eqnarray}\label{eq:ErgA12}
R_{A_{12}}(\v{p}_S)&=&\max\Big\{0, \frac{\omega_1}{Z_C}\Big(q^{H}_{10}-(q^{H}_{01}-q^{H}_{21})q^{C}_{10}-q^{C}_{20})\Big), 
\frac{\omega_1}{Z_C}\Big(q^{H}_{10}-(1-q^{H}_{10})-(1-q^{H}_{01}-q^{H}_{21})q^C_{10}\Big)\nonumber\\&+&  \frac{\omega_2}{Z_C}\Big((1-q^{H}_{10})+(1-q^{H}_{01}-q^{H}_{21})q^C_{10}-(q^{H}_{01}-q^{H}_{21})q^{C}_{10}-q^{C}_{20}\Big)\Big\},
\end{eqnarray}

\textbf{Calculation of the amount of transferred heat from Hot bath  and extractable ergotropy from $\v{p}_S$ using thermal process  $A_{13}$:} Thermal process $A_{13}$ transforms the initial probability vector given in Eq. \eqref{eq:Initialstate2} as follows 
\begin{eqnarray}
\label{eq:TransitionA13}
&&\frac{1}{Z_C}\begin{pmatrix}
0 & q^H_{01}-q^H_{21} & 1 \\
1-q^H_{20} & 1-q^H_{01}+q^H_{21} & 0 \\
q^H_{20} & 0 & 0 
\end{pmatrix}
\begin{pmatrix}
1 \\
q^C_{10} \\
q^C_{20} 
\end{pmatrix}=
\frac{1}{Z_C}\begin{pmatrix}
q^{C}_{10}(q^H_{01}-q^H_{21})+q^{C}_{20} \\
1-q^H_{20}+q^C_{10}(1-q^H_{01}+q^H_{21}) \\
q^H_{20} 
\end{pmatrix}=\text{Spec}(\hat{\rho}^6_{S}),\nonumber\\
\end{eqnarray}
From the table \ref{Tableofbetaorder}, we see that $\beta$-ordering of the final probability vector is given by (321).
It is straightforward to see that all ordering of the populations in the final probability vector given in Eq. \eqref{eq:TransitionA13} is compatible with this $\beta$- order. Thus ergotropy can be calculated as 
\begin{eqnarray}\label{eq:ErgA13}
R_{A_{13}}(\v{p}_S)&=&\max\Big\{0,\frac{\omega _2-\omega_1}{Z_C} \Big(2 q^H_{20}-1-q^C_{10} \Big(1+q^H_{21}-q^H_{01}\Big)\Big),\;\frac{\omega _1}{{Z_C}} \Big(1-q^C_{20}-q^H_{20}\nonumber\\&+&q^C_{10}\Big(1+2q^H_{21}-2q^H_{01}\Big)\Big),\; \frac{\omega_1}{Z_C}\Big(q^C_{10} \Big(1+q^H_{21}-q^H_{01}\Big)+1-2q^H_{20}\Big)+\frac{\omega_2}{Z_C}\Big(q^H_{20}-q^C_{20}\nonumber\\ &+&q^C_{10}q^H_{21}-q^C_{10}q^H_{01}\Bigg),\;\frac{\omega_1}{Z_C}\Big(1 - q^C_{20} - q^H_{20}+ q^C_{10}\Big(1 + 2 q^H_{21} - 2 q^H_{01}\Big)\Big) -\frac{\omega_2}{Z_C}\Big(1-2q^H_{20}\nonumber\\
&+&q^C_{10} \Big(1 + q^H_{21} - q^H_{01}\Big)\Big),\; \frac{\omega_2}{Z_C}\Big(q^H_{20}-q^C_{20}+q^C_{10}\Big(q^H_{21}-q^H_{01}\Big)\Big)\Big\}.
\end{eqnarray}

\bibliographystyle{vancouver}
\bibliography{sample}

\end{document}